\documentclass[abbrv,USenglish]{article}

\newcommand{\ifRRelse}[2]{#1}

\usepackage{multicol}
\usepackage{amsmath}
\usepackage{amsfonts}
\usepackage{amssymb}
\usepackage[vario]{fancyref}
\usepackage{verbatim}
\usepackage{fancybox}
\usepackage{wasysym}
\usepackage{enumerate}
\usepackage{float}
\usepackage{subfig}
\usepackage{qtree}
\usepackage{url}
\usepackage{longtable}
\usepackage{algorithm}
\usepackage{wrapfig}

\usepackage{answers}
\usepackage{eurosym}
\usepackage{xspace}
\usepackage{comment}
\usepackage{color}
\usepackage{multirow}
\usepackage{bussproofs}
\usepackage[T1]{fontenc}
\usepackage{graphicx}

\usepackage{ifpdf}
\ifpdf  
\else
  \usepackage{pstricks}
\fi

\newcommand{\skipthis}[1]{}

\def\colorize<#1>{\temporal<#1>{\color{black}}{\color{red}}{\color{black}}}
\newcommand{\bitem}{\begin{itemize}}
\newcommand{\eitem}{\end{itemize}}

\newcommand{\N}{\mathbb N}

\newcommand{\isdef}{\stackrel{\mbox{\tiny def}}{=}}

\newcommand{\trueform}{\top}

\newcommand{\iseq}{\simeq}
\newcommand{\niseq}{\not\simeq}

\newcommand{\istrivial}[1]{}

\newcommand{\true}{\texttt{true}}
\newcommand{\false}{\texttt{false}}

\newcommand{\A}{{\cal A}}

\newcommand{\C}{{\cal C}}

\newcommand{\E}{{\cal E}}

\newcommand{\X}{{\cal X}}
\newcommand{\Y}{{\cal Y}}
\newcommand{\V}{{\cal V}}

\newcommand{\M}{{\cal M}}

\newcommand{\U}{{\cal U}}
\renewcommand{\P}{{\cal P}}

\newcommand{\calZ}{{\cal Z}}
\newcommand{\vars}{{\cal V}} 

\newcommand{\arity}{\mathrm{\it arity}}

\newcommand{\dom}{\text{dom}}
\newcommand{\var}{\text{var}}

\newcommand{\select}{\text{\it select}}
\newcommand{\store}{\text{\it store}}

\newcommand{\ie}{i.e.\ }

\newcommand{\wrt}{w.r.t.\ }

\usepackage{natbib}

\newcommand{\sltf}{\mathrm{select}}
\newcommand{\slt}[1]{\sltf(#1)}
\newcommand{\storef}{\mathrm{store}}
\renewcommand{\store}[1]{\storef(#1)}
\newcommand{\SP}{${{\cal SA}}^{\prec}_{\sel}$\xspace}
\newcommand{\sel}{\text{\it sel}}

\newcommand{\elementary}{elementary\xspace}

\usepackage[normalem]{ulem}

\newtheorem{newtheo}{Theorem}
\newtheorem{newex}[newtheo]{Example}

\newtheorem{newdef}[newtheo]{Definition}
\newtheorem{newrem}[newtheo]{Remark}
\newtheorem{newprop}[newtheo]{Proposition}
\newtheorem{newlem}[newtheo]{Lemma}
\newtheorem{newcor}[newtheo]{Corollary}

  \newtheorem{assertion}[newtheo]{Assumption}

\ifRRelse{\newtheorem{proof}{Proof.}}{}

\title{A Superposition Calculus for Abductive Reasoning}

\author{M. Echenim and N. Peltier\ifRRelse{\\ Univ. Grenoble Alpes, CNRS, LIG\\ F-38000 Grenoble, France}{}}

\ifRRelse{}{\institute{Univ. Grenoble Alpes, F-38000 Grenoble, France \and
CNRS, LIG}}

\begin{document}

\ifRRelse{\date{}}{}
\maketitle

\ifRRelse{}{\email{Mnacho.Echenim@imag.fr, Nicolas.Peltier@imag.fr}




\date{}
}

\newcommand{\hypcst}{{\cal A}} 

\newcommand{\funcs}{{\cal F}}
\newcommand{\preds}{{\cal P}}

\newcommand{\variabless}{{\cal V}}
\renewcommand{\arity}{\text{\it ar}}
\newcommand{\head}{{\text \it head}}
\newcommand{\flatcl}[1]{\mathfrak{C}_{\text{\it flat}}(#1)}
\newcommand{\emptypos}{\varepsilon}

\newcommand{\almosteq}[1]{\sim^{#1}_{\hypcst}}
\newcommand{\red}[2]{{#1}_{\downarrow #2}}
\newcommand{\compl}{\mathrm{c}}

\ifRRelse{
\newcommand{\Mnacho}[1]{}
\newcommand{\MnachowithAnswer}[2]{}

\newcommand{\nikonew}[1]{}
\newcommand{\nikonewanswer}[2]{}

\newcommand{\nikoold}[1]{}
}
{
\newcommand{\Mnacho}[1]{{\color{blue} \textit{M: #1}}}
\newcommand{\MnachowithAnswer}[2]{{\color{blue} \textit{M: \sout{#1}}} {{\color{red} \textit{R: #2}}} }

\newcommand{\nikonew}[1]{{\color{red} #1}}
\newcommand{\nikonewanswer}[2]{{\color{red} #1}{\color{green} #2}}

\newcommand{\nikoold}[1]{{\color{red} Niko}: #1}
}
\newcommand{\set}[1]{\left\{#1\right\}}
 \newcommand{\setof}[2]{\left\{#1\,|\:#2\right\}}

\renewcommand{\true}{\top}
\renewcommand{\false}{\bot}

\begin{abstract}
  We present a modification of the Superposition calculus that is
  meant to generate consequences of sets of first-order axioms. This approach is proven to be sound and deductive-complete in the presence of redundancy elimination rules, provided the
  considered consequences are built on a given finite set of ground terms, represented by constant symbols.
 In contrast to other approaches, most existing results about the termination of the Superposition calculus can be carried over to our procedure. This ensures in particular that the calculus is terminating for many theories of interest to the SMT community.
  \end{abstract}

\nikonew{J'utilise une numérotation unique pour tous  les environnements (plus lisible). Globalement le papier semble un peu trop long. N'hésite pas à supprimer des trucs ou à simplifier les preuves (on pourrait peut-être se débarrasser de certaines propositions).}

\section{Introduction}

The verification of complex systems is generally based on proving the
validity, or, dually, the satisfiability of a logical formula. A
standard practice consists in translating the behavior of the system
to be verified into a logical formula, and proving that the negation
of the formula is unsatisfiable.  These formul{\ae} may be
domain-specific, so that it is only necessary to test the
satisfiability of the formula modulo some background theory, whence
the name \emph{Satisfiability Modulo Theories problems}, or \emph{SMT
  problems}.  If the formula is actually satisfiable, this means the
system is not error-free, and any model can be viewed as a trace that
generates an error. The models of a satisfiable formula can therefore
help the designers of the system guess the origin of the  errors
and deduce how they can be corrected; this is the main reason for example why state-of-the-art SMT solvers feature automated model building tools \citep*[see for instance][]{CLP}.
However, this approach is not always satisfactory. First, there is the risk of an information overkill: indeed, the generated
model may be very large and complex, and discovering the origin of the error may require a long and difficult analysis. Second, the model may be too specific, in the sense that it only
corresponds to one particular execution of the system and that
dismissing this single execution may not be sufficient to fix the
system.  Also, there are generally many
interpretations on different domains that satisfy the formula. In order to understand where the
error(s) may come from, it
is generally necessary to analyze all of these models and to identify common patterns.
 This leaves the user with the burden of having to infer the
general property that can rule out all the undesired behaviors.
A more useful and informative solution would be to
directly infer the missing axioms, or hypotheses, that can be added in order to ensure the unsatisfiability of the input formula. These axioms can be viewed as sufficient conditions ensuring that the system is valid. Such conditions must
be {\em plausible} and {\em economical}:
for instance, explanations that contradict the axioms of the considered theories are
obviously irrelevant.

In this paper, we present what is, to the best of our knowledge, a novel
approach to this debugging problem: we argue that rather than studying
one or several models of a formula, more valuable information can be extracted
from the properties that hold in \emph{all} the models of the formula.
For example, consider the theory of arrays, which is axiomatized as
follows \citep*[as introduced by][]{McC62}:
\begin{eqnarray}
  \forall x, z, v.\ \slt{\store{x, z, v}, z}& \iseq& v,\label{ar_1}\\
  \forall x, z, w, v.\ z \iseq w \vee \slt{\store{x, z, v}, w} &\iseq& \slt{x, w}.\label{ar_2}
  \end{eqnarray}
  \begin{figure}[t]
    \begin{center}
  \includegraphics[width=8cm]{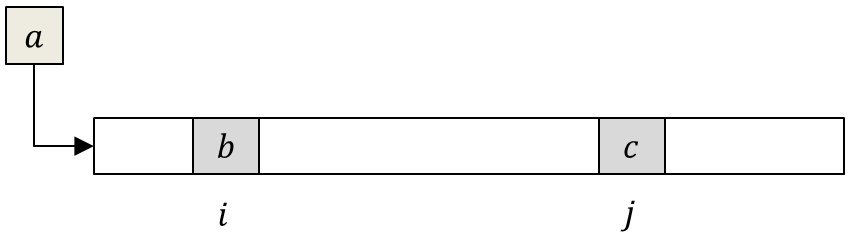}
  \end{center}
  \caption{Insertion into array $a$ of element $b$ at position $i$ and
  element $c$ at position $j$.}\label{fig:tablo}
  \end{figure}
  These axioms state that if element $v$ is inserted into array $x$ at
  position $z$, then the resulting array contains $v$ at position $z$,
  and the same elements as in $x$ elsewhere.  Assume that to verify
  that the order in which elements are inserted into a given array
  does not matter, the satisfiability of the following formula is
  tested (see also Figure \ref{fig:tablo}):
  \[\slt{\store{\store{a,i,b},j,c},k} \niseq
  \slt{\store{\store{a,j,c},i,b},k}.\] This formula asserts that there
  is a position $k$ that holds different values in the array obtained
  from $a$ by first inserting element $b$ at position $i$ and then
  element $c$ at position $j$, and in the array obtained from $a$ by
  first inserting element $c$ at position $j$ and then element $b$ at
  position $i$. It turns out that this formula is actually
  satisfiable, which in this case means that some hypotheses are
  missing. State-of-the-art SMT solvers such as Yices
  \citep*{DM06} or Z3 \citep*{DBLP:conf/tacas/MouraB08} can help find out what hypotheses are missing by outputting a model of
  the formula. In this case, Yices outputs \texttt{(= b 1) (= c 3)
    (= i 2) (= k 2) (= j 2)}, and for this simple example, such a
  model may be sufficient to quickly understand where the error comes
  from. However, a simpler and more natural way to determine what
  hypotheses are missing would be to have a tool that, when fed the
  formula above, outputs
  $i \iseq j \wedge b \niseq c$, stating that the formula can only
  be true when elements $b$ and $c$ are distinct, and are inserted at
  the \emph{same} position in array $a$. This information permits to know
  immediately what additional hypotheses must be made for the formula
  to be unsatisfiable. In this example, there are two possible
  hypotheses that can be added: $i \niseq j$ or $b\iseq c$.

  We investigate what information should be provided to
  the user and how it can be obtained, by distinguishing a set of
  ground terms on which additional hypotheses are allowed to be
  made. These terms may be represented by a particular set of constant symbols, called
  \emph{abducible constants} or
  simply \emph{abducibles}, and the
  problem boils down to determining what ground clauses containing
  only abducible constants are logically entailed by the formula under
  consideration, since the negation of any of these clauses can be
  viewed as a set of additional hypotheses that make the formula
  unsatisfiable. Indeed, by duality, computing implicants (or explanations) of a formula $\phi$
  is equivalent to computing implicates (i.e., logical consequences) of $\neg \phi$.
In order to compute such implicates, we devise a variant of the Superposition calculus \citep{BG94,DBLP:books/el/RV01/NieuwenhuisR01} that is deductive-complete for
the considered set of abducible constants, i.e., that can generate all the clauses built on abducible constants using finite set of predicate symbols including $\iseq$ that are logical consequences of the input clause set (up to redundancy). Our procedure is defined by enriching
the standard calculus with some new mechanisms allowing the assertion of relevant hypotheses during the proof search. These additional hypotheses are stored as constraints associated with the clauses and are propagated along the derivations. If the empty clause can be generated under a conjunction of hypotheses $\X$, then the conjunction of the original formula and $\X$ is unsatisfiable.
An essential feature of this approach is that the conditions are not asserted arbitrarily or eagerly, using a generate-and-test approach (which would be inefficient): instead they are {\em discovered} on a need basis, either by considering residual equations of unification failures (for positive hypotheses) or by
negating some of the literals occurring in the clauses (for negative hypotheses).


\subsection*{Related Work}

The generation of implicants (or, by duality, of implicates) of logical formul{\ae} has many applications in system verification and artificial intelligence, and this problem has been thoroughly investigated in the context of propositional logic. The earlier approaches use refinements of the resolution method \citep{tison1967generalization,kean1990incremental,de1992improved,simon2001efficient}, while more recent and more efficient  proposals use decomposition-based procedures \citep{jackson1990computing,henocque2002prime,matusiewicz2009prime,matusiewicz2011tri}.
These methods mainly focus on the efficient representation of information, and develop compact ways of storing and manipulating huge sets of implicates.

In contrast, the approaches handling abductive reasoning in first-order or equational logic are very scarce. Implicates can be generated automatically from sets of first-order clauses by using the resolution rule \citep{DBLP:conf/fair/Marquis91}. However, when dealing with equational clause sets, the addition of equality axioms leads to inefficiency and divergence in almost all but trivial cases.
\citet*{Knill92equalityand} use a proof technique called \emph{surface resolution} for generating implicates of Horn clauses in equational logic. The proposed approach, based on a systematic flattening of the terms and on the application of the resolution principle with substitutivity axioms, is very general and has some nice theoretical properties, but it is also very inefficient. The search space is huge, because the systematic abstraction of every subterm destroys all ordering or unifiability constraints, and termination is very rare.
\citet*{DBLP:journals/igpl/MayerP93}
describe a tableaux-based (or, dually, a sequent-based) proof procedure for abductive reasoning.
The intuitive idea is to apply the usual decomposition rules of propositional logic, and then compute
the {formul\ae} that force the closure of all open branches in the tableaux, thus yielding sufficient conditions ensuring unsatisfiability. The approach can be extended to first-order logic, by relying on reverse skolemization techniques in order to eliminate the Skolem symbols introduced inside the branches for handling existential quantifiers. Again, this approach is not well-suited for handling equality, and no termination results are presented.
\citet*{Tran10} show that the Superposition
calculus can be used to generate {\em positive} and {\em unit} implicates for some specific theories.
This approach is closer to ours, since it is based on the Superposition calculus, hence handles equality in an efficient way; however it is very focused: indeed, it is well-known that the Superposition calculus is not deductive-complete in general, for instance it cannot generate the clause $a \niseq b$ from the clause $f(a) \niseq f(b)$, although $f(a) \niseq f(b) \models a \niseq b$.

While the previous approaches rely  on usual complete proof procedures for first-order logic, more recent work builds on the recent developments and progresses in the field of Satisfiability Modulo Theories by devising algorithms relying on theory-specific decision procedures.
\citet*{DBLP:conf/cade/Sofronie-Stokkermans10,DBLP:conf/cade/Sofronie-Stokkermans13}  devises a technique for generating abductive explanations in local extensions of decidable theories. The approach reduces the considered problem to a formula in the basic theory by instantiating the axioms of the extension.
\citet*{DBLP:conf/cav/DilligDMA12} generate an incomplete set of  implicants of formul{\ae} interpreted in decidable theories by combining quantifier-elimination (for discarding useless variables) with model building tools (to construct sufficient conditions for satisfiability).
In contrast to these approaches, our method is proof-theoretic, hence it is generic and self-sufficient. The drawback is that it requires the adaptation of usual theorem provers instead of using them as black boxes (see also Example \ref{ex:paar} for a comparison of our method with the simplification technique devised by \citet{DBLP:conf/cav/DilligDMA12}).

\citet*{Wer13} proposes a method to derive abductive explanations from first-order logical programs, under several distinct non-classical semantics, using a reduction to second-order quantifier-elimination. Both the considered framework and the proposed techniques completely depart from our work.

\subsection*{Organization of the Paper}

The rest of the paper is structured as follows.
In Section \ref{sect:prel} we review basic definitions and  adapt usual results to our framework.
In Section \ref{sect:sup} the new Superposition calculus \SP is presented, and it is shown in Section \ref{sect:comp} that it is deductive-complete for ground clauses built on the set of abducible constants. In Section \ref{sect:ref} some refinements of the calculus are presented, aiming at more efficiency. In Section \ref{sect:term}, we show that most termination results holding for the usual Superposition calculus also apply to \SP.
The present paper is a thoroughly expanded and  revised version of \citep{EP12a}. See Section \ref{sect:ep12} for more details on the relationship of \SP with the calculus in \citep{EP12a}. \nikonew{au lieu de ``previous calculus''}

\section{Preliminaries}

\label{sect:prel}

\newcommand{\cmin}{c_0}

\subsection{Basic Definitions}

\label{sect:basic}

\nikonew{modifs, symboles de prédicat (plus clair). J'ai hésité à autoriser des prédicats explicitement et à rajouter la résolution tu me diras si la solution te convient.}

The set of \emph{terms} is built as usual on a set of \emph{function symbols} $\funcs$ including a set of \emph{predicate symbols} $\preds$, containing in particular a special constant $\true$, and a set of \emph{variables} $\vars$. Every symbol $f \in \funcs$ is mapped to a unique {\em arity} $\arity(f) \in \N$. The set $\funcs_n$ is the set of function symbols of arity $n$; an element of $\funcs_0$ is a \emph{constant}. \nikonew{modif} A term whose head is in $\preds$ is {\em boolean}.

An \emph{atom} (or {\em equation}) is an unordered pair of terms, written $t \iseq s$, where $t$ and $s$ are terms.
A \emph{literal} is either an atom or the negation of an atom (i.e., a {\em disequation}), written $t \niseq s$. For every literal $l$, we denote by $l^{\compl}$ the complementary literal of $l$, which is defined as follows: $(t \iseq s)^{\compl} \isdef t\niseq s$ and
$(t \niseq s)^{\compl} \isdef t \iseq s$. We use the notation $t \bowtie s$ to denote a literal of the form $t \iseq s$ or $t \niseq s$, and $t \not\bowtie s$ then denotes the complementary literal of $t \bowtie s$.
As usual, a non-equational atom $p(\vec{t})$ where $p \in \preds$ is encoded as an equation $p(\vec{t}) \iseq \true$. For readability, such an equation is sometimes written $p(\vec{t})$, and $p(\vec{t}) \niseq \true$ can be written $\neg p(\vec{t})$.
A \emph{clause} is a finite multiset of literals, sometimes written as a disjunction. The empty clause is denoted by $\Box$.  For technical reasons, we assume that the predicate symbols only occur in atoms of the form $t \iseq \true$, where $t \not = \true$ (literals of the form $\true \niseq \true$ can be removed from the clauses and clauses containing a literal $\true \iseq \true$ can be dismissed; equations of the form $p(\vec{t}) = q(\vec{s})$ with $p,q \in \preds \setminus \{ \true \}$ are forbidden).
\nikonew{modifs -- il faudrait dire aussi que le calcul n'engendre pas de telles équations, mais ça semble assez évident.}
For every clause $C = \{ l_1,\ldots,l_n\}$, $C^{\compl}$ denotes the set of unit clauses
$\{ \{ l_i^{\compl} \} \mid i \in [1,n] \}$ and for every set of unit clauses $S = \{ \{ l_i \} \mid i \in [1,n] \}$, $S^{\compl}$ denotes the clause
$\{ l_1^{\compl},\ldots,l_n^{\compl} \}$.
Throughout the paper, we assume that $\prec$ denotes some  fixed reduction ordering on terms \citep[see, e.g.,][]{baader1998}
such that $\true \prec t$, for all terms $t \not = \true$, extended to atoms, literals and clauses as usual\footnote{The literals $t \iseq s$ and $t \niseq s$ are ordered as $\{ \{ t \}, \{ s \} \}$ and $\{ \{ t,s \} \}$, respectively.}.

The set of variables occurring in an expression (term, atom, literal, clause) $E$ is denoted by $\var(E)$. If $\var(E)= \emptyset$ then $E$ is {\em ground}. A \emph{substitution} is a function mapping variables to terms. For every term $t$ and for every substitution $\sigma$, we denote by $t\sigma$ the term obtained from $t$ by replacing every variable $x$ by its image w.r.t.\ $\sigma$. The {\em domain} of a substitution is the set of variables $x$ such that $x\sigma \not = x$. A substitution $\sigma$ is {\em ground} if for every $x$ in the domain of $\sigma$, $x\sigma$ is ground.

A \emph{position} is a finite sequence of  positive integers. A position $p$ \emph{occurs} in a term $t$ if either $p = \emptypos$ or if $t = f(t_1,\ldots,t_n)$, $p = i.q$ with $i \in [1,n]$ and $q$ is a position in $t_i$.
If $p$ is a position in $t$, the terms $t|_{p}$ and $t[s]_p$ are defined as follows: $t|_{\emptypos} \isdef t$,
$t[s]_{\emptypos} \isdef s$, $f(t_1,\ldots,t_n)|_{i.q} \isdef (t_i)|_q$ and $f(t_1,\ldots,t_n)[s]_{i.q} \isdef f(t_1,\ldots,t_{i-1},t_i[s]_q,t_{i+1},\ldots,t_n)$.

\newcommand{\eflat}[1]{#1-flat}

\nikonew{Modifs.}
Given a set of constants $E$, a literal $t \bowtie s$ is {\em \eflat{$E$}} if either
$t,s \in \vars \cup E$ or $t = p(t_1,\dots,t_n)$, $s = \true$ and $t_1,\dots,t_n \in \vars \cup E$. A clause is \emph{\eflat{$E$}} if all its literals are {\eflat{$E$}}.
The set of {\eflat{$E$}} clauses is denoted by $\flatcl{E}$.
A clause is {\em flat} if it is {\eflat{$\funcs_0$} and {\em \elementary} if it is \eflat{$\hypcst$} and contains no symbol in $\preds$ (in other words, every literal is of the form $a \bowtie b$ with $a,b \in \vars \cup \hypcst$).

\newcommand{\dominates}{\sqsubseteq}

\newcommand{\Aset}{$\hypcst$-set\xspace}

An {\em interpretation} is a congruence relation on ground terms. An interpretation $I$ {\em validates} a clause $C$ if for all ground substitutions $\sigma$ of domain $\var(C)$ there exists $l \in C$ such that
 either $l = (t \iseq s)$ and $(l,s)\sigma \in I$, or
 $l = (t \niseq s)$ and $(l,s)\sigma \not \in I$.

\subsection{Abducible Constants and $\hypcst$-Sets}

\label{sect:aset}

In this section we introduce the notion of an $\hypcst$-set, that provides a convenient way of representing partial interpretations defined on a particular set of constant symbols.
Let $\hypcst \subseteq \funcs_0$ be a set of constants, called the  \emph{abducible constants}.
The set $\hypcst$ is fixed by the user and contains all constants on which the abducible \MnachowithAnswer{abducible?}{ok} formul{\ae} can be constructed. 
We assume that $f(\vec{t}) \succ a$, for all $a \in \hypcst$ and $f \not \in \hypcst$, and that
$q(t_1,\dots,t_n) \succ p(a_1,\dots,a_n)$ if $a_1,\dots,a_n \in \hypcst$, $p,q$ are predicate symbols and $\exists i \in [1,n]\, t_i \succ a_1,\dots,a_n$.

\nikonew{Modifs}

\begin{newdef}
An {\em \Aset} is a set of \eflat{$\hypcst$} literals $\X$ satisfying the following properties.
\begin{itemize}
\item{If $L \in \X$ and $L$ is not ground then $L$ is negative or of the form $p(t_1,\ldots,t_n) \iseq \true$.}
\item{If $\set{L[a]_p, a \iseq b} \subseteq \X$, where $a,b \not = \true$, then $L[b]_p \in \X$.}
\item{$a \iseq a \in \X$, for all $a \in \hypcst$.}
\end{itemize}
\MnachowithAnswer{décalé la définition de complet ici}{ok}
 An \Aset $\X$ is \emph{positive} if it only contains positive literals, and {\em complete} if for every ground \eflat{$\hypcst$} atom $A$, $\X$ contains either $A$ or $\neg A$.
\end{newdef}
Note that all elementary \MnachowithAnswer{\elementary?}{oui} positive literals in $\X$ must be ground whereas negative or non elementary literals possibly contain variables. \MnachowithAnswer{zapper ça? a positive \Aset is necessarily ground}{oui}.
Informally,  a satisfiable \Aset can be viewed as a partial interpretation on the constant symbols in $\hypcst$. The positive elementary \MnachowithAnswer{\elementary?}{oui} literals in $\X$ define an equivalence relation between elements on $\hypcst$ and the negative elementary \MnachowithAnswer{elementary?}{oui} literals specify the equivalence classes that are known to be distinct. Literals of the form $p(t_1,\dots,t_n) \bowtie \true$ specify the interpretation of predicate symbols on constants of $\hypcst$. Variables correspond to unknown (or unspecified) constant symbols in $\hypcst$. Complete {\Aset}s are total interpretations on $\A$.

This definition of {\Aset}s is given for theoretical purposes only: in practice, they can be more conveniently represented by a set of oriented equations of the form $\{ a_i \iseq b_i \mid i \in [1,n] \}$, where $\forall i \in [1,n]\ a_i,b_i \in \hypcst$, $a_i \succ b_i$ and $i \not = j \Rightarrow a_i \not = a_j$, together with a set of irreducible
literals of the form $c \not \iseq d$ or $p(c_1,\dots,c_n) \bowtie \true$, where
$\forall i \in [1,n]$, $c,d,c_1,\dots,c_n \not = a_i$.
When convenient, we may represent an \Aset by a set $X$ of equations and disequations, with the intended meaning that we are actually referring to the smallest \Aset $\X$ that contains $X$.

\begin{newex}
Let $\hypcst = \set{a,b,c,d}$ and $x\in \V$. Then the set \[\X = \{ a \iseq a,\, b\iseq b,\, c \iseq c,\, d \iseq d,\, e\iseq e,\, a \iseq b,\, c \not \iseq a,\, c \not \iseq b,\, d \not \iseq x\}\] is an \Aset.
Assuming an ordering such that $a \succ b \succ c \succ d$, it can be more concisely represented by
$\{ a \iseq b, c \not \iseq b, d \not \iseq x \}$.
$\X$ defines a partial interpretation in which $a,b$ are known to be equal and distinct from $c$, while $d$ is distinct from some unspecified constant $x$ ($x$ can represent $a,b,c$ or $e$ -- if $x$ represents $d$ then the set is unsatisfiable). The interpretation is only partial since
 it can be extended into a total interpretation that satisfies either $a\iseq d$ or $a\niseq d$.
\end{newex}

\begin{newdef}
\label{def:red}
  For every \Aset $\X$ and for every expression (term, atom, literal,
  clause or clause set) $E$, we denote by $\red{E}{\X}$ the expression
  obtained from $E$ by replacing every constant $a \in \hypcst$ in $E$
  by the smallest (according to $\prec$) constant $b$ in $\hypcst$
  such that $a \iseq b \in \X$.  We write $t \almosteq{\X} s$ iff
  $\red{t}{\X} = \red{s}{\X}$ and $t \almosteq{} s$ iff there exists
  an \Aset $\X$ such that $t \almosteq{\X} s$.
   This definition is extended to substitutions: we write $\sigma = \red{\theta}{\X}$ if $x\sigma = \red{(x\theta)}{\X}$ and $\sigma\almosteq{\X} \theta$ if for all $x \in \dom{(\sigma)} \cup \dom{(\theta)}$, $x\sigma \almosteq{\X} x\theta$.
\end{newdef}

\begin{newprop}\label{prop:flatred}
  Let $C$ be a clause, $\sigma$ be a substitution and $\X$ be an \Aset. If $\red{(C\sigma)}{\X}$ is \eflat{$\hypcst$} \MnachowithAnswer{ajout}{ok} (resp. \elementary), then so is $C$.
\end{newprop}

\begin{proof}
  The contrapositive is obvious: if $C$ is not \eflat{$\hypcst$}, then it contains a non-boolean term $t$ that is not in $\A \cup \V$. But then, neither  $t\sigma$ nor $\red{t\sigma}{\X}$ can be in $\A\cup \V$, and $\red{(C\sigma)}{\X}$ cannot be \eflat{$\hypcst$}. \nikonew{ajout} The reasoning is similar for \elementary clauses.
\end{proof}

\subsection{$\hypcst$-Unification}

\newcommand{\Asubstitution}{$\hypcst$-substitution\xspace}
\newcommand{\Aunifier}{$\hypcst$-unifier\xspace}
\newcommand{\Aunification}{$\hypcst$-unification\xspace}
\newcommand{\Aunifiable}{$\hypcst$-unifiable\xspace}

\newcommand{\moregeneral}{\geq_{\hypcst}}
\newcommand{\equivto}{\sim_{\hypcst}}

$\hypcst$-unification is an extension of unification that, given two terms $t$ and $s$, aims at computing a substitution $\sigma$ such that $t\sigma \almosteq{} s\sigma$, meaning that $t\sigma$ and $s\sigma$ are equal up to a renaming of constants in $\hypcst$. The set of necessary constant renamings is collected and stored in a positive \Aset. This set corresponds exactly to residual (non-solvable) equations obtained when applying standard unification algorithms.
\begin{newex}
The terms $f(a,b)$ and $f(x,x)$ are not unifiable in the standard sense, but they are \Aunifiable. The substitution $\sigma: \{ x \mapsto a \}$ is an \Aunifier of these two terms, together with the \Aset $\{ a\iseq a,\, b\iseq b,\, a \iseq b \}$.
\end{newex}
\begin{newdef}
An {\em \Asubstitution} is a pair $(\sigma,\X)$ where $\sigma$ is a substitution and $\X$ is an
 \Aset containing only equations between elements of $\hypcst$.
An \Asubstitution $(\sigma,\X)$ is an {\em \Aunifier} of an equation $t \iseq s$  iff $t\sigma \almosteq{\X} s\sigma$.
Two terms admitting an \Aunifier are {\em \Aunifiable}.
\end{newdef}

Intuitively, if $(\sigma,\X)$ is an {\Aunifier} of an equation $t \iseq s$, then the equations in $\X$ can be used to reduce $t$ and $s$ to terms that are unifiable in the standard sense.

\begin{newdef}
\label{pseudounif:gen}
An \Asubstitution $(\sigma,\X)$ is \emph{more general} than an \Asubstitution $(\sigma',\X')$, written $(\sigma,\X) \moregeneral (\sigma',\X')$, if there exists a (standard) substitution $\theta$ such that the two following conditions hold:
\begin{itemize}
\item{$\X \subseteq \X'$.}
\item{For every $x \in \vars$, $x\sigma'\almosteq{\X'} x\sigma\theta$.}
\end{itemize}
We write $(\sigma,\X) \equivto (\sigma',\X')$ if $(\sigma,\X) \moregeneral (\sigma',\X')$ and
$(\sigma',\X') \moregeneral (\sigma,\X)$.
\end{newdef}

\begin{newex}
  Let $\A = \set{a,b,c}$, and consider the following substitutions and {\Aset}s:
  \[\begin{array}{rclcrcl}
    \sigma & = & \set{x\mapsto a,\ y\mapsto c,\ z \mapsto f(a,z')} & \textrm{ and } & \X & = & \set{a\iseq c}\\
   \sigma' & = & \set{x\mapsto a,\ y\mapsto b,\ z\mapsto f(b,b),}  & \textrm{ and } & \X' & = & \set{a\iseq b,\ b\iseq c}.
 \end{array}\]
 By letting $\theta = \set{z'\mapsto b}$, it is simple to verify that $(\sigma,\X) \moregeneral (\sigma',\X')$.
\end{newex}

Note that most general {\Aunifier}s are not unique modulo variable renamings. For example, the equation $f(g(a),g(b)) \iseq f(g(x),g(y))$ admits several most general unifiers, including $(\{ x \rightarrow a, y \rightarrow b \}, \{ a \iseq b \})$,  $(\{ x \rightarrow b, y \rightarrow a \}, \{ a \iseq b \})$, \ldots which are of course all $\equivto$-equivalent.
{\Aunifier}s can be computed by a slight adaptation of the usual unification algorithm (see Appendix \ref{ap:unif} for details).


\section{\hypcst-Superposition Calculus}

\label{sect:sup}

\newcommand{\SPR}[1]{${{\cal SA}}^{\prec}_{\sel}(#1)$\xspace}
\newcommand{\SPP}{${{\cal SAR}}^{\prec}_{\sel}$\xspace}

\newcommand{\Sup}[2]{${{\cal SP}}^{#1}_{#2}$\xspace}

\newcommand{\ccl}[2]{[#1\!\mid\! #2]}

\newcommand{\Aclause}{$\hypcst$-clause\xspace}

\newcommand{\alwayssgreater}[1]{\succ_{\hypcst}}
\newcommand{\alwaysgreater}[1]{\succeq_{\hypcst}}
\newcommand{\alwaysssmaller}[1]{\prec_{\hypcst}}
\newcommand{\alwayssmaller}[1]{\preceq_{\hypcst}}

In this section we define an extension of the standard Superposition calculus \citep{BG94,DBLP:books/el/RV01/NieuwenhuisR01} with which it is possible to generate all {\eflat{$\hypcst$} implicates of a considered clause set. The calculus handles constrained clauses, called \emph{{\Aclause}s}, the constraint part of an \Aclause being an \Aset containing all the equations and disequations needed to derive the corresponding non-constraint part from the original clause set. Unification is replaced by \Aunification, and the \Aset of the generated \Aunifier is appended to the constraint of the conclusion of the rule. Furthermore, an additional inference rule, called the {\em $\hypcst$-Assertion} rule, is introduced in order to add  disequations to the constraints.

\begin{newdef}
An {\em \Aclause} is a pair $\ccl{C}{\X}$ where $C$ is a clause and $\X$ is an \Aset. If $\X = \emptyset$, then we may write $C$ instead of $\ccl{C}{\emptyset}$.
 \end{newdef}

 In what follows, we first define the ordering and selection function the calculus is based upon before presenting the inference rules and redundancy criterion of the $\hypcst$-Superposition calculus. We conclude this section by showing that the calculus is sound.

\subsection{Ordering and Selection Function}

 We begin by introducing some additional notations and terminology.

\begin{newdef}
 For all terms $t$, $s$, we write $t \alwayssgreater{\X} s$ if for every
\Aset $\X$ and ground substitution $\sigma$,
we have $\red{t\sigma}{\X} \succ \red{s\sigma}{\X}$.
This ordering is extended to atoms, literals and clauses in a similar way to $\prec$.
\end{newdef}

Intuitively $t \alwayssgreater{\X} s$ means that
$t$ is always greater than $s$, regardless of the names of the constants in $\hypcst$.

\begin{newex}
If $a,b,c \in \hypcst$ and $f(x) \succ a \succ b \succ c$, then we have $f(b) \alwaysgreater{\X} a$, but $f(a) \not \alwaysgreater{\X} f(b)$, since $\red{f(a)}{\{a \iseq c \}} = f(c) \prec f(b) = \red{f(b)}{\{ a \iseq c \}}$.
\end{newex}

\begin{newdef}
A substitution $\sigma$ is {\em $\X$-pure} if for all variables $x \in \var(\X)$, $x\sigma$ is either a variable or a constant in $\hypcst$.
\end{newdef}

 \begin{newdef}
 A function $\sel$ is a \emph{selection function} for an ordering $>$ iff $\sel$ maps every clause $C$ to a set of literals in $C$ such that $\sel(C)$  either contains a negative literal or
contains all literals that are $>$-maximal in $C$.
 \end{newdef}

We consider a selection function  $\sel$ for the ordering $\alwaysgreater{}$, that satisfies the following assumptions.
\begin{assertion}
The function $\sel$ is stable under {\em \Asubstitution}s, i.e.,  for every clause $C$, for every literal $l \in C$ and for every \Asubstitution $(\eta,\X)$, if $\red{l\eta}{\X} \in \sel(\red{C\eta}{\X})$, then $l \in \sel(C)$.
\end{assertion}

\begin{assertion}
\label{hyp:B}
For every \Aclause $C$, if $\sel(C)$ contains a literal of the form $p(\vec{t}) = \true$ then $\sel(C)$ contains no negative literal of the form $a \niseq b$ with $a,b \in \vars \cup \hypcst$.
\end{assertion}
Assumption \ref{hyp:B} can always be fulfilled since negative literals can be selected arbitrarily.

\label{cond:sel}

\subsection{Inference Rules}

The calculus \SP\ is defined by the  rules below.
The standard Superposition calculus (denoted by \Sup{\prec}{\sel}) coincides with \SP if $\hypcst = \emptyset$.

\begin{newrem}
Following our convention, in all rules, if $\X$, $\Y$ are two {\Aset}s, then $\X \cup \Y$ does not denote the mere union of $\X$ and $\Y$, but rather the smallest \Aset containing both $\X$ and $\Y$ (it is obtained by transitive closure from the union of $\X$ and $\Y$). For example, if $\set{a,b,c} \subseteq \hypcst$ with $a\succ b\succ c$,
$\X = \{ a \iseq a, b \iseq b, c \iseq c, a \iseq b \}$ and
$\Y = \{ a \iseq a, b \iseq b, c \iseq c, a \iseq c \}$, then
$\X \cup \Y$ denotes the \Aset
 $\{ a \iseq a, b \iseq b, c \iseq c, a \iseq b, a \iseq c, b \iseq c \}$.
 Similarly, if $\X$ is an \Aset and $\sigma$ is an $\X$-pure substitution, then $\X\sigma$ denotes the smallest \Aset containing $\X\sigma$. For instance, if
 $\X = \{ a \iseq a, b \iseq b, a \iseq b, x \not \iseq y \}$ and $\sigma = \{ x \mapsto a \}$, then $\X\sigma = \{ a \iseq a, b \iseq b, a \iseq b, a \not \iseq y, b \not \iseq y \}$.
\end{newrem}

\subsubsection*{$\hypcst$-Superposition}

\begin{center}
\begin{minipage}{0.8\linewidth}
\[
\begin{tabular}{c}
$\ccl{C \vee t \bowtie s}{\X},\quad \ccl{D \vee u \iseq v}{\Y}$ \\
\hline
$\ccl{C \vee D \vee t[v]_p \bowtie s}{\X \cup \Y \cup \E}\sigma$
\end{tabular}
\]
If $\bowtie \in \{ \iseq, \not \iseq \}$, $(\sigma,\E)$ is an $(\X\cup \Y)$-pure
most general \Aunifier of $u$ and $t|_{p}$,
$v\sigma \not \alwaysgreater{\Y\sigma} u\sigma, s\sigma \not \alwaysgreater{\X\sigma} t\sigma$,
$(t \bowtie s)\sigma \in \sel((C \vee t \bowtie s)\sigma)$,
$(u \iseq v)\sigma \in \sel((D \vee u \iseq v)\sigma)$ and  if $t|_{p}$ is  a variable then $t|_{p}$ occurs in $\X$.
\end{minipage}
\end{center}

We shall refer to the left and right premises of the inference rule as the {\em into} and {\em from} premises, respectively.
The main difference with the usual Superposition rule (besides the replacement of $\succ$ by $\alwayssgreater{}$ and of unifiers by {\Aunifier}s) is that superposition into a variable is permitted, provided the considered variable occurs in the constraint part of the clause. The reason is that these do not actually represent variables in the usual sense, but rather placeholders for (unknown) constants (see also Example \ref{ex:vars}).

  By definition of the calculus, variables can only occur in the constraints if the $\hypcst$-Assertion rule (see below) is applied on a non-ground literal. This is the case because, by definition of \Aunification, the other rules  add only ground equations into the constraints. Furthermore, by definition, a non-ground literal can be added to the constraints only if the considered clause is variable-eligible, \ie contains a selected literal of the form $x \iseq t$, where $x \not \prec t$.  This cannot happen if the clause set is variable-inactive  \citep{ABRS09}. However, there exist theories of interest that are not variable-inactive, for instance the theory of arrays with axioms for constant arrays
 (e.g., $\forall x, \select(t,x) \iseq c$).
\MnachowithAnswer{est-ce qu'il ne faudrait pas faire une remarque sur le cas où on a des prémisses de la forme $p(t) \iseq \top$ et $p(t)\niseq \bot$?}{J'ai rajouté la remarque ci-dessous, dis moi si c'est à ça que tu pensais.}

Note that the rule applies if $t$ and $u$ are of the form $p(\vec{t}) \bowtie \trueform$ and
$p(\vec{s}) \iseq \trueform$ (with $p = \emptypos$), in which case $t[v]_p \bowtie s$ is of the form $\true \bowtie \true$. If $\bowtie$ is $\iseq$ then the \Aclause is a tautology and can be deleted, and if $\bowtie$ is $\niseq$ then the literal $\true \niseq \true$ is deleted from the clause as explained before. The rule is essentially equivalent to Ordered Resolution in this case \citep*[see for instance][]{LEI96}.

\subsubsection*{$\hypcst$-Reflection}
\begin{center}
\begin{minipage}{0.8\linewidth}
\[
\begin{tabular}{c}
$\ccl{C \vee t \not \iseq s}{\X}$ \\
\hline
$\ccl{C}{\X \cup \E}\sigma$
\end{tabular}
\]
If $(\sigma,\E)$ is an $\X$-pure most general \Aunifier of $t$ and $s$ and $(t \not \iseq s)\sigma \in \sel((C \vee t \not \iseq s)\sigma)$.
\end{minipage}
\end{center}

\subsubsection*{Equational $\hypcst$-Factorization}

\newcommand{\eqfact}{equational\xspace}
\newcommand{\Eqfact}{Equational\xspace}

\begin{center}
\begin{minipage}{0.8\linewidth}
\[
\begin{tabular}{c}
$\ccl{C \vee t \iseq s \vee u \iseq v}{\X}$ \\
\hline
$\ccl{C \vee s \not \iseq v \vee t \iseq s}{\X \cup \E}\sigma$
\end{tabular}
\]

If $(\sigma,\E)$ is an  $\X$-pure most general \Aunifier of $t$ and $u$, $s\sigma \not \alwaysgreater{\X\sigma}
t\sigma$, $v\sigma \not \alwaysgreater{\X\sigma} u\sigma$ and $(t \iseq s)\sigma \in \sel((C
\vee t \iseq s \vee u \iseq
v)\sigma)$.
\end{minipage}
\end{center}

For technical convenience, we assume that $s \not \iseq v$ is omitted in the conclusion if $s\sigma = v\sigma$.



 \newcommand{\pp}{\beta}

\newcommand{\projection}[2]{\Phi(#1,#2)}
\newcommand{\projectionbis}[2]{\Psi(#1,#2)}
\newcommand{\projectiontype}[3]{\Phi_{#3}(#1,#2)}


\subsubsection*{$\hypcst$-Assertion}

\begin{center}
\begin{minipage}{0.8\linewidth}
\[
\begin{tabular}{c}
$\ccl{t \iseq s \vee C}{\X}$ \\
\hline
$\ccl{C}{\X \cup \{ t \not \iseq s \}}$
\end{tabular}
\]

If $t,s \in \hypcst \cup \V$, $t \iseq s \in \sel(t \iseq s \vee C)$ and $\hypcst \not = \emptyset$.
\end{minipage}
\end{center}

\begin{center}
\begin{minipage}{0.8\linewidth}
\[
\begin{tabular}{c}
$\ccl{p(t_1,\dots,t_n) \bowtie \true \vee C}{\X}$ \\
\hline
$\ccl{C}{\X \cup \{ p(t_1,\dots,t_n) \not \bowtie \true \}}$
\end{tabular}
\]

If $t_1,\dots,t_n \in \hypcst \cup \V$, $p(t_1,\dots,t_n) \bowtie \true \in \sel(t \iseq s \vee C)$ and $\hypcst \not = \emptyset$.
\end{minipage}
\end{center}

\subsubsection*{$\hypcst$-Substitutivity Rule}

\begin{center}
\begin{minipage}{0.9\linewidth}
\[
\begin{tabular}{c}
$\ccl{t_1 \iseq s_1 \vee C_1}{\X_1} \dots \ccl{t_n \iseq s_n \vee C_n}{\X_n}$ \\
\hline
$\ccl{p(t_1,\dots,t_n) \bowtie \true \vee C_1 \vee \dots \vee C_n}{\{  p(s_1,\dots,s_n) \bowtie \true \} \cup \bigcup_{i=1}^n \X_i }$
\end{tabular}
\]
\MnachowithAnswer{A revérifier (je n'ai pas relu les démos), mais pourquoi on n'a pas les $\X_i$ dans les contraintes?}{oubli de ma part}
\MnachowithAnswer{If $p$ is a predicate symbol of arity $n$, $x_1,\dots,x_n$ are variables. quid?}{deleted, de même que le pb précédent, c'est un relicat d'une version précédente de la règle qui introduisait juste des axiomes de la forme $\ccl{p(\vec{x})}{\{ p(\vec{x}) \}}$ }
\end{minipage}
\end{center}
The rule can be applied also by replacing some of the premisses $\ccl{t_i \iseq s_i \vee C_i}{\X_i}$ by variants of the Reflexivity axiom $x \iseq x$ (note that if all premisses are of this form then the conclusion is a tautology).

\subsection{Soundness}

The interpretation of an \Aclause is defined as a logical implication:

\begin{newdef}
An interpretation $I$ {\em validates} an \Aclause $\ccl{C}{\X}$ iff for every $\X$-pure ground substitution $\sigma$ of domain $\var(C) \cup \var(\X)$,
 either $I \not \models \X\sigma$ or $I \models C\sigma$.
\end{newdef}

If $I \models \ccl{C}{\X}$ for all interpretations $I$, then
$\ccl{C}{\X}$ is a \emph{tautology}.\MnachowithAnswer{attention: on se sert de la notion de tautologie dans la définition de la redondance. Est-ce qu'il ne faudrait pas intervertir 'soundness' et 'redundancy'? D'ailleurs la proposition ci-dessous est uniquement utilisée dans l'annexe. Pourquoi ne pas la décaler là-bas?}{c'est le plus simple en effet} In particular, this property holds if
$\X$ is unsatisfiable,  if $\red{C}{\X}$ contains two complementary literals or a literal of the form $t \iseq t$, or if all the literals in $C$ occurs in $\X$.

\begin{newtheo}
Let $S$ be a set of {\Aclause}s. If $C$ is generated from $S$ by one of the rules of \SP then $S \models C$.
\end{newtheo}
\begin{proof}
  It suffices to prove that all the rules are sound, i.e., that the
  conclusion of the rule is a logical consequence of the
  premises. This is due to the fact that if $(\sigma,\E)$ is an
  \Aunifier of $t \iseq s$, then the \Aclause $\ccl{t\sigma \iseq
    s\sigma}{\E}$ is valid in all interpretations. Then the proof follows
  by a straightforward inspection of the rules, as in the usual case.
\end{proof}

\subsection{Redundancy}

\newcommand{\redundant}{$\hypcst$-redundant\xspace}
\newcommand{\quasipositive}{quasi-positive\xspace}

We now adapt the standard redundancy criterion to {\Aclause}s. An \Aclause is {\em \quasipositive} if the only negative literals occurring in it are of the form $p(\vec{t}) \niseq \true$.
\MnachowithAnswer{reformulation de la def ci-dessous}{oui}
\begin{newdef}\label{def:ared}
\label{def:red_crit}
An \Aclause $\ccl{C}{\X}$ is {\em \redundant} in a set of {\Aclause}s $S$ if either $\ccl{C}{\X}$ is a tautology, or for every ground substitution $\theta$ of the variables in $\ccl{C}{\X}$ such that $\X\theta$ is a satisfiable \Aset, one of the following conditions hold.
\begin{itemize}
\item{There exists an \Aclause $\ccl{D}{\Y}$ and a substitution $\sigma$ such that
$D\sigma \subseteq C\theta$ and $\Y\sigma \subseteq \X\theta$.}
\item{
    If $\hypcst = \emptyset$ or $C\theta$ is not both {\eflat{$\hypcst$}} and {\quasipositive}, then there exist
{\Aclause}s
$\ccl{D_i}{\Y_i}$ and substitutions $\sigma_i$ ($1 \leq i \leq n$), such that:
\begin{itemize}
\item $\Y_i\sigma_i \subseteq \X\theta$ for all $i = 1, \ldots, n$,
\item $\X\theta, D_1\sigma_1,\ldots,D_n\sigma_n \models C\theta$,
\item $C\theta \alwaysgreater{\Y} D_1\sigma_1,\ldots,D_n\sigma_n$.
\end{itemize}}


\end{itemize}
\end{newdef}
When applied to standard clauses (with $\hypcst=\emptyset$), this notion of redundancy coincides
with the usual criterion \citep[see for instance][]{BG94,DBLP:books/el/RV01/NieuwenhuisR01}.\MnachowithAnswer{à quoi correspond 'citep'? Je trouve bof de mettre le 'see for instance' après les références mais je ne sais pas si c'est modifiable...}{erreur latex de ma part, la syntaxe est soit {\tt citep[apr\`es]\{refs\}} soit {\tt citep[avant][apr\`es]\{refs\}} }

It is easy to check that the standard redundancy detection
rules such as subsumption, tautology deletion or equational
simplification, are particular cases of this redundancy criterion.
\nikonew{ajout remarque:}
Note that the second item in Definition \ref{def:red_crit}
is similar to the usual redundancy criterion of the Superposition calculus \citep[see, e.g,][]{BG94}, with the following differences:
(i) the entire constraint part of the considered \Aclause may be used to infer the clausal part, disregarding any ordering condition, (ii) the condition only applies to clauses that are not both \eflat{$\hypcst$} and \quasipositive. For the clauses that are \eflat{$\hypcst$} and \quasipositive, redundancy testing is limited to tautology deletion and subsumption (this is necessary to ensure completeness, see Remark \ref{rem:red}).

\begin{newex}
Let $\hypcst = \set{a,b,c}$. The \Aclause $\ccl{a \not \iseq c \vee b \not \iseq c \vee f(x)\iseq d}{a \not \iseq b}$ is \redundant in any set $S$,
since for all ground substitutions $\theta$, $a \not \iseq b \models (a \not \iseq c \vee b \not \iseq c\vee f(x)\iseq d)\theta$.

The \Aclause $\ccl{f(a,b) \iseq c \vee g(a) \iseq d}{a \niseq b}$ is
\redundant in $\{ f(a,x) \iseq c \vee a \iseq b \}$. Indeed, let
$\sigma = \{ x \mapsto b \}$, then $a \niseq b, f(a,x)\sigma \iseq c
\vee a \iseq b \models f(a,b) \iseq c \vee g(a) \iseq d$ and $f(a,b)
\iseq c \vee g(a) \iseq d \alwaysgreater{} f(a,x)\sigma \iseq c \vee a
\iseq b$.
 \end{newex}

The following result is a straightforward consequence of Definition \ref{def:ared}.

\begin{newprop}\label{prop:instred}
  If $\ccl{C}{\X}$ is redundant in a set $S$, then for any \Asubstitution $(\sigma, \Y)$, $\ccl{C\sigma}{\X\cup \Y}$ is also redundant in $S$.
\end{newprop}


\begin{newdef}
A set $S$ is {\em \SP-saturated}
if every {\Aclause} that can be derived from {\Aclause}s in $S$ by a rule in
\SP is redundant in $S$.
\end{newdef}

\subsubsection*{Examples}

\nikonew{J'ai déplacé cette section ici, parce qu'on se servait de notion satisfiabilité etc.}

We provide simple application examples.

\begin{newex}
  Let $S = \{ g(f(x)) \iseq d,\, f(a) \iseq a,\, g(b) \iseq b,\, d
  \iseq c \}$.  Assume that $\hypcst = \{ a,b,c \}$.  By applying the
  $\hypcst$-Superposition rule on the terms $f(x)$ and $f(a)$, we
  derive the clause $g(a) \iseq d$ (note that this application of the
  rule is equivalent to the usual one). Then the
  $\hypcst$-Superposition rule can be applied again on the terms $g(a)$ and
  $g(b)$. The unification yields the constraints $a \iseq b$, hence
  the following \Aclause is derived: $\ccl{b \iseq d}{a \iseq b}$.
  The Assertion rule cannot be applied on $b \iseq d$, since this
  literal is not \eflat{$\hypcst$}. Instead,the application of the
  $\hypcst$-Superposition rule on the term $d$ (note that we must have
  $d \succ b,c$ since $d \not \in \hypcst$ and $b,c \in
  \hypcst$) yields: $\ccl{b \iseq c}{a \iseq b}$. Finally, the
  Assertion rule can be applied on $b \iseq c$ since this literal is
  \eflat{$\hypcst$}, thus generating $\ccl{\Box}{b \not \iseq c \wedge a
    \iseq b}$. This \Aclause is equivalent to the clause $b \iseq c
  \vee a \not \iseq b$, and we have $S \models b \iseq c \vee a \not
  \iseq b$.
\end{newex}

The second example involves predicate symbols.

\begin{newex}
We consider two  functions $f$ and $g$ such that $f$ and $y \mapsto g(x,y)$ are increasing, together with abducible constants $a$, $b$, $i$ and $j$. The aim is to determine under which conditions the property $f(g(a,i)) \leq f(g(b,j))$ holds. The problem is formalized as follows (where $t \leq s$ stands for $(t \leq s) \iseq \true$ and $< \in\preds$, $x,y,u \in \vars$):
$S = \{ x \not \leq y \vee f(x) \leq f(y), x \not \leq y \vee g(u,x) \leq g(u,y), f(g(a,i)) \not \leq f(g(b,j)) \}$. For conciseness, the axioms corresponding to $\leq$ (e.g., transitivity) are omitted since they play no rôle in our context.

The Superposition rule applies on the first and last clauses, yielding $g(a,i) \not \leq g(b,j)$. Then the rule applies again from the latter clause into the second one, and it generates:
$\ccl{i \not \leq j}{\{ a \iseq b \}}$. Finally the $\hypcst$-Assertion rule yields the \Aclause:
$\ccl{\Box}{\{ i \leq j, a \iseq b \}}$, meaning that the desired property is fulfilled if
$i \leq j$ and $a \iseq b$ hold.
\end{newex}

\MnachowithAnswer{j'intervertirais cet exemple avec le précédent.}{on pourra en discuter, mais personnellement cela me paraît + logique comme ça (d'abord des exemples montrant le fonctionnement, puis des exemples rentrant dans le détail du calcul et montrant la nécessité des nouvelles règles)}

The $\hypcst$-Assertion rule is necessary to guarantee deductive completeness, as shown in the example below.

\begin{newex}
  Consider the (satisfiable) clause set: $S \isdef \{ y \iseq x \vee
  f(x,x,y) \iseq a, f(a,b,c) \not \iseq a\}$, where $\hypcst \isdef \{
  a,b,c \}$ and $x,y$ are variables. It is simple to verify that
  $S\models a\niseq b\vee c\iseq a$, and the calculus is designed to
  generate from $S$ a clause of the form $\ccl{\Box}{\X}$, where $\neg
  \X \equiv a\niseq b\vee c\iseq a$. In order to generate such a
  clause, it is clear that one has to unify $f(x,x,y)$ and $f(a,b,c)$, since
   the unification of $f(a,b,c)$ and $a$ leads to an
  immediate failure, so that the Reflection rule is not
  applicable. This is feasible only if the condition $a \iseq
  b$ is added to the constraints of the obtained clause, yielding a constrained
  clause of the form: $\ccl{c \iseq a}{a \iseq b}$. The literal $c\iseq a$ is deleted using the $\hypcst$-Assertion rule, by appending the disequation $c\niseq a$ to the constraints, thus obtaining the required \Aclause: $\{ \ccl{\Box}{a \iseq b,\, c \not \iseq a} \}$. 
\end{newex}

The last example shows that the $\hypcst$-Substitutivity rule is also needed for completeness.

\begin{newex}
\label{ex:subst}
Consider the clause set: $S \isdef \{ a \iseq b \}$.
It is clear that $S \models p(a) \iseq \true \vee p(b) \niseq \true$ for any predicate symbol $p$ of arity $1$, but
$\ccl{\Box}{\{ p(a) \niseq \true, p(b) \iseq \true \}}$ cannot be generated
without
the help of the $\hypcst$-Substitutivity rule.
The above implicate is indeed obtained as follows: The $\hypcst$-Substitutivity rule generates the \Aclause $\ccl{p(x) \iseq \true}{p(x) \iseq \true}$, then the $\hypcst$-Superposition rule applies from $a \iseq b$, yielding
$\ccl{p(a) \iseq \true}{p(b) \iseq \true}$, and the desired result is obtained by applying the $\hypcst$-Assertion rule.
Note that the equation $p(a) \iseq p(b)$ does not need to be inferred in our context since predicate symbols are allowed only in atoms of the form $t \iseq \true$.
Considering implicates built on arbitrary function symbols (with nested applications) would lead to divergence since, e.g., an infinite number of clauses of the form $f^n(a) \iseq f^n(b)$ (with $n \in \N$) could be derived from the above clause.
\end{newex}

\nikonew{ajout:}

\begin{newrem}
\label{rem:red}
The previous example also shows the importance of the restriction on the redundancy criterion.
Indeed, if the criterion is relaxed by removing the condition ``$C\theta$ is not \eflat{$\hypcst$} and \quasipositive'' in the second item of Definition \ref{def:red_crit}, then the \Aclause $\ccl{p(a) \iseq \true}{p(b) \iseq \true}$ is redundant in $S$ (since $a \iseq b \prec p(a) \iseq \true$ and
$a \iseq b, p(b) \iseq \true \models p(a) \iseq \true$).
Consequently no non redundant inferences apply on $S$ and the implicate
$p(a) \iseq \true \vee p(b) \niseq \true$ cannot be generated.
\end{newrem}

\section{Deductive Completeness}

\label{sect:comp}

We show in this section that
\SP is deduction-complete for the clauses in $\flatcl{\hypcst}$. More precisely, we  prove that for any \SP-saturated set $S$ and clause $C\in \flatcl{\hypcst}$, if $S\models C$ then $S$ contains an \Aclause of the form $\ccl{\Box}{\Y}$ where $C^{\compl} \models \Y$.
The result is obtained in the following way. Given such a set $S$ and clause $C$, we consider the smallest \Aset $\X$ that contains $C^{\compl}$, and construct a set of standard ground clauses $\projection{S}{\X}$ such that:
\begin{itemize}
\item{$\projection{S}{\X}$ contains all ground instances of clauses in $S$, as well as a set of unit clauses equivalent to $\X \equiv C^{\compl}$.}
\item{$\projection{S}{\X}$ is saturated under a slightly adapted version  of the Superposition calculus which is refutationally complete.}
\end{itemize}
Since $S \cup \set{C^{\compl}}$ is unsatisfiable and the considered calculus is refutationally complete, these two
properties together will entail that $\projection{S}{\X}$ contains the
empty clause.  Finally,
we show that this is possible
only if $S$ contains an \Aclause of the required form.

First, we formally define the notions of $\hypcst$-implicates and prime $\hypcst$-implicates.

\newcommand{\imp}[1]{I_{\hypcst}(#1)}
\newcommand{\abdcl}[1]{\mathcal{C}_{\hypcst}(#1)}

\begin{newdef}
Let $S$ be a set of {\Aclause}s.
A clause $C$ is an {\em $\hypcst$-implicate} of $S$ if it satisfies the following conditions.
 \begin{itemize}
 \item{$C$ is $\A$-flat and ground. 
     }
 \item{$C$ is not a tautology.}
\item{$S \models C$.}
\end{itemize}
 $C$ is a {\em prime $\hypcst$-implicate} of $S$ if, moreover, $C \models D$ holds for every $\hypcst$-implicate $D$ of $S$ such that $D \models C$.
 We denote by $\imp{S}$  the set of $\hypcst$-implicates of $S$.
\end{newdef}

\begin{newdef}
We denote by $\abdcl{S}$ the set of clauses  of the form $(\X\sigma)^{\compl}$, where
$\ccl{\Box}{\X} \in S$ and $\sigma$ maps each variable $x$ in $\X$ to some constant symbol $a \in \hypcst$ in such a way that $\X\sigma$ is satisfiable\footnote{In other words, $\sigma$ is such that for every $u \not \iseq v \in \X$, $u\sigma\not = v\sigma$.}.
We write $S \dominates S'$ if for every clause $C' \in S'$, there exists
$C \in S$ such that $C \models C'$.
\end{newdef}

Our goal is to prove that $\abdcl{S} \dominates \imp{S}$ when $S$ is \SP-saturated, i.e., that
every prime implicate of $S$ occurs in $\abdcl{S}$ (up to equivalence).

\newcommand{\selB}{\sel_{\Phi}}

\subsection{Definition of $\projection{S}{\X}$}

\newcommand{\ff}{\alpha}

\newcommand{\superp}[2]{\text{\it sup}(#1,#2)}

Let $\ff$ and $\pp$ be two arbitrarily chosen function symbols  not occurring in $S$,
where $\arity(\ff) = 1$ and $\arity(\pp) = 0$. We assume that $\forall a \in \hypcst, \pp \succ \ff(a)$ and that $\forall g \not \in \hypcst \cup \{ \ff \}, g(\vec{t}) \succ \pp$.\Mnacho{peut-être pas important, mais quelle comparaison avec prédicats?} \nikonew{je n'ai pas compris ta question. Il faudra qu'on en discute. Si un prédicat $p$ n'est pas dans $\hypcst$ -- en particulier si l'arité de $p$ est non nulle -- alors $\forall a \in \hypcst\, p(\dots) \succ a$.}

\nikonew{Note: ajout de $\cup \{ \ff \}$ dans les contraintes sur $g$ ci-dessus, sinon évidemment les contraintes sont inconsistantes}

For every clause $C$ and clause set $S$,
$\superp{C}{S}$ denotes the set inductively defined as follows.
\begin{itemize}
\item{$C \in \superp{C}{S}$.}
\item{If $D \in \superp{C}{S}$ and $D'$ is obtained by applying the standard Superposition rule into $D$ from a positive and \elementary clause in $S$, then $D' \in \superp{C}{S}$.}
\end{itemize}
A clause set $S$ is {\em non-redundant} iff
for every clause $C \in S$, $C$ is not redundant in $S \setminus \{ C \}$.
For every clause set $S$, it is easy to obtain a non-redundant subset of $S$ that is equivalent to $S$ by recursively removing from $S$ every clause $C$ that is redundant in $S \setminus \{ C \}$.

We define the set of standard ground clauses $\projection{S}{\X}$ as
well as a selection function $\selB$ as follows.

\begin{newdef}
\label{projext:def}
Let $S$ be a set of {\Aclause}s and let $\X$ be an \Aset.
We denote by $\projection{S}{\X}$ the set
$$\projection{S}{\X} \isdef \projectiontype{S}{\X}{\ref{projext:set}} \uplus
\projectiontype{S}{\X}{\ref{projext:eq}} \uplus
\projectiontype{S}{\X}{\ref{projext:noeq}} \uplus
\projectiontype{S}{\X}{\ref{projext:pred}} \uplus
\projectiontype{S}{\X}{\ref{projext:pp}}$$ where for $i = 1,\ldots,
5$, $\projectiontype{S}{\X}{i}$ is defined as follows:
\begin{enumerate}
\item{$\projectiontype{S}{\X}{\ref{projext:set}}$ is the set of
    clauses of the form $\red{D\sigma}{\X} \vee C'$, where
    $\ccl{D}{\Y}\in S$, $\sigma$ is a ground substitution of domain
    $\var(D)$ such that $\Y\sigma \subseteq \X$ and $\red{x\sigma}{\X} = x\sigma$ for all  $x \in \var(D)$,  and $C'$ is defined
    as follows:
    \begin{itemize}
    \item{$C' \isdef \Box$ if $D\sigma$ is {\eflat{$\hypcst$}} and \quasipositive; }
    \item{$C' \isdef (\pp \niseq \true)$ otherwise.}
    \end{itemize}
    The selection function $\selB$ is defined on
    $\projectiontype{S}{\X}{\ref{projext:set}}$ as follows:
    $\selB(\red{D\sigma}{\X} \vee C')$ contains all literals
    $\red{l}{\X}$ such that $l \in \sel(D\sigma)$ and one of the
    following holds:
    \begin{itemize}
    \item $l$ is negative,
    \item $\sel(D\sigma)$ is positive and $\red{l}{\X}$ is
      $\succ$-maximal in $\red{D\sigma}{\X} \vee C'$.
    \end{itemize}
   \label{projext:set}}

\item{$\projectiontype{S}{\X}{\ref{projext:eq}}$ is the set of unit clauses of the form $c \iseq \red{c}{\X}$, where $c \in \hypcst$ and $c \not = \red{c}{\X}$. The selection function is defined on $\projectiontype{S}{\X}{\ref{projext:eq}}$ by: $\selB(c \iseq \red{c}{\X}) \isdef \set{c \iseq \red{c}{\X}}$.\label{projext:eq}}

\item{$\projectiontype{S}{\X}{\ref{projext:noeq}}$ is the set of non-redundant clauses in 
    \[\bigcup_{a\niseq b \in \X}\superp{\ff(\red{a}{\X}) \not \iseq
      \ff(\red{b}{\X})}{\projectiontype{S}{\X}{\ref{projext:set}}},\]
    and for all $C \in
    \projectiontype{S}{\X}{\ref{projext:noeq}}$, $\selB(C)$ contains all negative literals in $C$.
\label{projext:noeq}}

\item{$\projectiontype{S}{\X}{\ref{projext:pred}}$ is the set of non-redundant clauses in \MnachowithAnswer{remplacé $f$ par $p$}{oui} \MnachowithAnswer{Pourquoi est-ce qu'on ne réduit pas les constantes $a_i$?}{oubli de ma part} 
    \[\bigcup_{p(a_1,\dots,a_n) \bowtie \true \in \X}\superp{p(\red{a_1}{\X},\dots,\red{a_n}{\X}) \bowtie \true}{\projectiontype{S}{\X}{\ref{projext:set}}},\]
    and for all $C \in
    \projectiontype{S}{\X}{\ref{projext:noeq}}$, $\selB(C)$ contains all literals of the form $t \bowtie \true$ in $C$. Note that the symbol $\bowtie$ occurring in the generated clause is the same as the one in the corresponding literal $p(a_1,\dots,a_n) \bowtie \trueform$ of $\X$. \MnachowithAnswer{il faut voir si ça ne prête pas à confusion, la notation $\bowtie$ ici: on peut s'imaginer que si on a $p\niseq \true$ dans $\X$, alors on engendre des clauses en se servant de $p\iseq \true$ dans le sup.}{J'ai rajouté la phrase précédente, dis-moi si cela te convient. Autres options: donner les deux définitions avec $\iseq$ ou $\niseq$, ou écrire quelque chose comme:
    $\bigcup_{l \in \X, l = p(a_1,\dots,a_n) \bowtie \true}\superp{\red{l}{\X}}{\projectiontype{S}{\X}{\ref{projext:set}}}$}
\label{projext:pred}}

\item{$\projectiontype{S}{\X}{\ref{projext:pp}} = \{ \pp \iseq \true\} \cup \{ \ff(u) \niseq \ff(v) \vee u \iseq v \mid u,v \in \hypcst, u = \red{u}{\X}, v = \red{v}{\X}, u \not = v \}$. We let
$\selB(\pp \iseq \true) \isdef \set{\pp \iseq \true}$, and
$\selB(\ff(u) \niseq \ff(v) \vee u \iseq v) \isdef \set{\ff(u) \niseq \ff(v)}$. \label{projext:pp}}
\end{enumerate}
It is easy to verify that the sets $\projectiontype{S}{\X}{i}$ with $i=1,\ldots,5$ are disjoint.
The \emph{type} of a clause $C \in \projection{S}{\X}$ is the number $i$ such that
$C \in  \projectiontype{S}{\X}{i}$.
\end{newdef}

\begin{newex}\label{ex:projection}
  Let $\A = \set{a,b,c,d,e}$, and $\X$ be the \MnachowithAnswer{reflexive-}{oui} reflexive-transitive closure of $\{ a \iseq
  b, c \iseq d, b \not \iseq e \}$, where $a \succ b \succ c \succ d
  \succ e$.  Consider the set of clauses
  \[S\ =\ \{ f(a) \iseq c \vee a \niseq b,\  b\niseq c,\ c\iseq d,\  \ccl{g(x,y) \iseq f(d)}{y \not
    \iseq e},\  \ccl{f(x) \iseq x}{a \iseq c} \}.\]
  Then
  $\projection{S}{\X}$ is decomposed as follows:
\begin{description}
\item[$\projectiontype{S}{\X}{\ref{projext:set}}$:] {This set consists
  of $f(b) \iseq d \vee a \niseq b \vee \pp \niseq \true$, $b \niseq d \vee \pp \niseq
  \true$, $d\iseq d$ and $g(t,b) \iseq f(d) \vee \pp \niseq \true$,
  where $t$ ranges over the set of all ground terms. The constants $a$
  and $c$ occurring in $S$ are respectively replaced by $b =
  \red{a}{\X}$ and $d = \red{c}{\X}$ in $\projection{S}{\X}$. The
  \Aclause $\ccl{f(x) \iseq x}{a \iseq c}$ generates no clauses in
  $\projection{S}{\X}$, since $(a \iseq c) \not \in \X$.}
\item[$\projectiontype{S}{\X}{\ref{projext:eq}}$:]{$\set{a \iseq b, c \iseq d}$. }
\item[$\projectiontype{S}{\X}{\ref{projext:noeq}}$:]{$\set{\ff(b) \not \iseq \ff(e), \ff(d) \not \iseq \ff(e)}$. The first clause is constructed from $(b \not \iseq e) \in \X$, the second one is generated by Superposition into $\ff(b) \not \iseq \ff(e)$ from the clause $b \iseq d$ above.}
\item[$\projectiontype{S}{\X}{\ref{projext:pred}}$:]{$\emptyset$. There is no predicate symbols other than $\iseq$.}
\item[$\projectiontype{S}{\X}{\ref{projext:pp}}$:] This set consists of the following clauses: \[\set{\pp \iseq \true,\ \ff(b) \niseq \ff(d) \vee b \iseq d,\ \ff(b) \niseq \ff(e) \vee b \iseq e,\ \ff(d) \niseq \ff(e) \vee d \iseq e}.\]
\end{description}
\end{newex}

\begin{newrem}
  The addition of $\ff$ is irrelevant from a semantic point of view,
  since \MnachowithAnswer{petite modif}{oui (j'ai remplacé une occurrence de $a$ par $b$)} by construction, $\ff(a) \niseq \ff(b)$ if and only if $a\niseq b$ for all $a,b \in \hypcst$; it is
  possible to replace all atoms of the form $\ff(x) \niseq \ff(y)$ by $x \niseq
  y$.  However, this technical trick ensures that all the clauses of
  type \ref{projext:noeq} are strictly greater than all \elementary \MnachowithAnswer{elementary?}{oui}
  {\eflat{$\hypcst$}} clauses in $\projection{S}{\X}$, which plays a
  crucial r\^ole in the proof of Lemma \ref{lem:saturated}.
  Similarly, the addition of the literal $\pp \not \iseq \true$ does
  not affect the semantics of the clause set (since by definition $\pp
  \iseq \true$ occurs in this set), but ensures that all
  clauses of type \ref{projext:set} that are not \quasipositive are strictly greater than all
  clauses of type \ref{projext:eq} or \ref{projext:noeq}.
\end{newrem}

\newcommand{\fixed}{superposable\xspace}

\begin{newprop}
  For all sets of clauses $S$ and {\Aset}s $\X$, $\selB$ is a
  selection function for the ordering $\succ$.
\end{newprop}

\begin{proof}
We must check that for every clause $C \in \projection{S}{\X}$, $\selB(C)$ contains either a negative literal in $C$ or all $\succ$-maximal
literals in $C$ (see Definition \ref{projext:def} for the notations).
This is immediate for  clauses of type \ref{projext:eq} and \ref{projext:pp}, since $\selB(C)=C$.
For  clauses of type \ref{projext:noeq}, we observe that $C$ necessarily contains a negative  literal, obtained from the literal 
$\ff(\red{a}{\X})\not \iseq \ff(\red{b}{\X})$ by Superposition. 
Similarly, all {\Aclause}s of
 type   \ref{projext:pred} contains a (unique) literal of the form \MnachowithAnswer{remplacé $p$ par $f$. Je le signale au cas où il y aurait une raison pour laquelle le symbole $f$ est utilisé}{ok, pas de pb} $p(a_1,\dots,a_n) \bowtie \true$, that is necessarily maximal.
Now assume that $C$ is a clause of type \ref{projext:set}, i.e., that $C = \red{D\sigma}{\X}\vee D''$ for some $\ccl{D}{\Y}$ in $S$ and $D'' \subseteq \set{\pp \niseq \true}$. If we suppose that $\selB(C)$ contains no
negative literal, then the same must hold for $\sel(D\sigma)$, thus $\sel(D\sigma)$ necessarily
    contains all $\alwaysgreater{}$-maximal literals in $D\sigma$, and by Assumption \ref{hyp:B}, \nikonew{ajout} if $D\sigma$ is \eflat{$\hypcst$} then it must be \quasipositive, and in this case $D'' = \Box$. Furthermore, by definition of $\alwaysgreater{}$,
    for all $m \in D\sigma$, if $\red{m}{\X}$ is $\succ$-maximal
    in $\red{D\sigma}{\X}$, then $m$ is $\alwaysgreater{}$-maximal in $D\sigma$, which entails that $\selB(C)$ contains all $\succ$-maximal
    literals in $C$ (note that if $D'' \not = \Box$ then $D$ is not \eflat{$\hypcst$}, hence $D\sigma \succ D''$).
\end{proof}

\MnachowithAnswer{ajout pour factorisation}{oui, ajout compléments}

\begin{newprop}\label{prop:equivps}
  Let $S_{init}$ be a set of standard clauses and let $S$ be a set of clauses generated from $S_{init}$ by \SP. Then $\projection{S}{\X} \models S_{init} \equiv S$.
\end{newprop}

\begin{proof}
  Let $S' = \projection{S}{\X}$ and consider the set of standard clauses $S_{cl}$ occurring in $S$, i.e., $S_{cl} \isdef \setof{C}{\ccl{C}{\emptyset} \in S}$. \nikonew{``Since \SP is sound, $S_{init} \equiv S \equiv S_{cl}$.'' Je pense qu'il faut en dire un peu plus ici car la soundness ne suffit pas: on pourrait penser que les clauses de $S_{init}$ peuvent être  supprimées car redondantes par rapport à des clauses non standard de $S$, dans ce cas on aurait pas équivalence entre $S_{cl}$ et $S_{init}$. Ce serait le cas par exemple si on pouvait transférer des littéraux de la partie clausale vers les contraintes. Par exemple si $S_{init} = \{ p \}$, on pourrait avoir $S = \{ \ccl{\Box}{\neg p} \}$ et $S_{cl} = \emptyset$. J'ai remis ce que j'avais mis dans la preuve précédente.}
   Since \SP is sound, $S_{init} \models S$. Furthermore, if a standard clause is \redundant in a set of {\Aclause}s, then it is also redundant w.r.t.\ the standard clauses in this set, by definition of the redundancy criterion. Thus $S_{cl} \equiv S \equiv S_{init}$.

  By construction, $S'$ contains all the clauses that can be obtained from ground instances of clauses in $S_{cl}$, by replacing every constant $a$ by $\red{a}{\X}$ and possibly adding literals of the form $\pp\niseq\true$. Since $S'$ contains all atoms of the form $a\iseq \red{a}{\X}$ where $a\neq \red{a}{\X}$ as well as the atom $\pp\iseq \true$, we deduce that $S'\models S_{cl}$, and that $S'\models S_{init} \equiv S$.
\end{proof}

\subsection{Saturatedness of $\projection{S}{\X}$}
The next lemma states that $\projection{S}{\X}$ is saturated w.r.t.\ a slight restriction of
the usual Superposition calculus.
We shall also use a refined version of the redundancy criterion.
\begin{newdef}
A set of ground clauses $S$ is {\em weakly saturated}
\wrt an inference rule in \Sup{\prec}{\selB} if every application of the rule on a set of premises $\{ C_1,\ldots,C_n \} \subseteq S$ (with $n =1,2$)
yields a clause $C$ such that there exists $\{ D_1,\ldots,D_m \} \subseteq S$ with $\forall i \in [1,m]$, $D_i \prec \max_{\prec}(\{ C_1,\ldots,C_n \})$ and $\{ D_1,\ldots,D_m \} \models C$.
\end{newdef}
\begin{newlem}
\label{lem:usualcomp}
Let $S$ be a set of ground clauses that is weakly saturated \wrt all rules in  \Sup{\prec}{\selB}. The set
 $S$ is satisfiable iff it does not contain $\Box$.
 \end{newlem}
 \begin{proof}
See \citep{BG94} or
\citep[][theorem 4.8]{DBLP:books/el/RV01/NieuwenhuisR01}.\MnachowithAnswer{donner la référence exacte? Il me semble que c'est le théorème 4.8}{done mais je n'arrive pas trouver la version finale de BG94 \smiley }
 \end{proof}

Lemma \ref{lem:saturated} below is the main technical result that is used to prove the completeness of \SP.

 \begin{newlem}
\label{lem:saturated}
Let $S$ be an \SP-saturated set of {\Aclause}s and let $\X$ be a ground and satisfiable
\Aset.
The set $\projection{S}{\X}$ is weakly saturated under all inference rules in
\Sup{\prec}{\selB}, except for \Eqfact Factorization on positive {\eflat{$\hypcst$}} clauses.
\end{newlem}
\begin{proof}
The proof is given in Appendix \ref{ap:saturated}.
\end{proof}

\begin{newrem}
\label{rem:fact}
The set $\projection{S}{\X}$ is not saturated under \Eqfact Factorization, because the literal $\pp \niseq \true$
is not added to the clauses that are positive and {\eflat{$\hypcst$}}, and such clauses can have non-positive
descendants.
For example, $\{ a \iseq b \vee a \iseq c, b \not \iseq c \}$ is \SP-saturated,
but
$\projection{S}{\emptyset} = \{ a \iseq b \vee a \iseq c, b \not \iseq c \vee \pp \niseq \true, \pp \iseq \true \}$ is not.
\end{newrem}

\newcommand{\spc}[1]{\mathrm{P^+}(#1)}

\begin{newcor}
\label{cor:unsat}
Let $S$ be an \SP-saturated set of {\Aclause}s and let $\X$ be a ground and satisfiable
\Aset.
If $\projection{S}{\X}$ is unsatisfiable then it contains $\Box$.
\end{newcor}
\begin{proof}
The proof is not straightforward since $\projection{S}{\X}$ is not saturated w.r.t. Equational Factorization on \eflat{$\hypcst$} {\Aclause}s, as explained above.
However it can be shown that the application of this rule on \eflat{$\hypcst$} is useless in our context; this is due to the fact that the constants in $\hypcst$ are not ordered (see Appendix \ref{ap:unsat} for details).
\end{proof}

\subsection{Deductive Completeness Theorem}
The previous results lead to the following theorem, which states that the calculus \SP can be used to generate all  ground implicates built on $\hypcst$.
\begin{newtheo}
\label{theo:comp}
Let $S_{init}$ be a set of standard clauses and let $S$ be a set of {\Aclause}s obtained from $S_{init}$ by \SP-saturation. Then $\abdcl{S} \dominates \imp{S_{init}}$.\MnachowithAnswer{est-ce qu'on ne pourrait pas mettre $\imp{S_{init}}$? C'est équivalent et ça me semble mieux.}{done}
\end{newtheo}
\begin{proof}
  \Mnacho{réécriture}
Let $C \in \imp{S_{init}}$, let $\X$ be the smallest \Aset containing $C^{\compl}$
 and
let $S' \isdef \projection{S}{\X}$. Note that $\X$ is ground since $C$ is ground.  Since $\X$ is
equivalent to $C^{\compl}$ and $C$ is not a tautology, this {\Aset} is satisfiable.
We first prove that $S'$ is equivalent to $S'' \isdef S \cup C^{\compl} \cup \{ \pp \iseq \true,
   \ff(u) \niseq \ff(v)
   \vee u \iseq v \mid u,v \in \hypcst \}$, and therefore unsatisfiable. By Proposition \ref{prop:equivps} \nikonew{remplacement de $S$ par $S_{init} \equiv S$} $S' \models S_{init} \equiv S$; since $\projectiontype{S}{\X}{\ref{projext:eq}} \cup \projectiontype{S}{\X}{\ref{projext:noeq}} \subseteq S'$, we have $S'\models C^\compl$, and since $\projectiontype{S}{\X}{\ref{projext:pp}} \subseteq S'$, we conclude that $S'\models S''$.
We now show that $S''$ entails all clauses in $S'$.
   \begin{description}
   \item [Clauses in $\projectiontype{S}{\X}{\ref{projext:pp}}$.] All the clauses in $\projectiontype{S}{\X}{\ref{projext:pp}}$ are in $S''$, and the result is obvious.
   \item [Clauses in $\projectiontype{S}{\X}{\ref{projext:eq}}$.] For all $c\in \hypcst$, $C^\compl \models c\iseq \red{c}{\X}$. Since $C^\compl \subseteq S''$, we have the result.
   \item [Clauses in $\projectiontype{S}{\X}{\ref{projext:set}}$.] Let $\ccl{D}{\Y} \in S$, and consider a ground substitution $\sigma$ such that $\Y\sigma \subseteq \X$ and $\red{x\sigma}{\X} = x\sigma$ for all  $x \in \var(D)$. Then $S''\models D\sigma$, and since $\Y\sigma \subseteq \X \equiv C^\compl$, $S''\models \red{(D\sigma)}{\X}$. But $\pp\iseq\true \in S''$, thus $S'' \models \red{(D\sigma)}{\X}\vee C'$, regardless of whether $C' = \Box$ or $C' = (\pp\niseq \true)$.
  \item [Clauses in $\projectiontype{S}{\X}{\ref{projext:pred}}$.] These clauses are all in $C^\compl$, hence the result is obvious.
   \item [Clauses in $\projectiontype{S}{\X}{\ref{projext:noeq}}$.] Consider a literal $a\niseq b \in \X$. Since the clause $a\iseq b \vee \ff(a) \niseq \ff(b)$ occurs in $S''$, we have $S''\models \ff(a)\niseq \ff(b)$; therefore, $S''\models \superp{\ff(\red{a}{\X}) \not \iseq
      \ff(\red{b}{\X})}{\projectiontype{S}{\X}{\ref{projext:set}}}$.
   \end{description}


Since $S''$ is unsatisfiable by construction, so is $S'$ and by Corollary \ref{cor:unsat}, $S'$
contains the empty clause.
This means that $S$ must contain
an \Aclause of
the form $\ccl{\Box}{\Y}$
where $\Y\theta \subseteq C^{\compl}$. By definition $\abdcl{S}$ contains the clause
$(\Y\theta)^{\compl}$ and since $\Y\theta \subseteq C^{\compl}$ we have $(\Y\theta)^{\compl} \models C$.
\end{proof}

Note that Theorem \ref{theo:comp} does not hold if $S$ is not obtained by \SP-saturation from a set of standard clauses; this is due to the fact that no inference is performed on the literals occurring in  the constraints.
For example, the set: $S = \{ \ccl{\Box}{a \iseq b}, \ccl{\Box}{a \niseq b} \}$ is clearly unsatisfiable and \SP-saturated, however we have $\abdcl{S} = \{ a \iseq b, a \niseq b \} \not \dominates \imp{S}$, since $\Box\in \imp{S}$.
We also provide an example showing that the theorem 
does not hold if $\hypcst$-Superposition into the variables occurring in the constraints is not allowed.
\begin{newex}
\label{ex:vars}
Let $S \isdef \{ x \iseq a \vee x \iseq c, x \iseq b \vee x \iseq d \}$ and $C \isdef e \iseq a \vee e \iseq b \vee c \iseq d$.
It is straightforward to verify that $S \models C$. The only way of generating an \Aclause $\ccl{\Box}{\X}$ such that
$\X\sigma \models C^{\compl}$ is to apply the Superposition rule on the literals $x \iseq c$ and $x \iseq d$ upon the term $x$, which is usually forbidden. This can be done by first applying the $\hypcst$-Assertion rule on the literals $x \iseq a$ and $x \iseq b$, yielding
$\ccl{x \iseq c}{\{ x \niseq a \}}$ and $\ccl{x \iseq d}{\{ x \niseq b \}}$. Then it is possible to apply the Superposition on the term $x$ since it occurs in the constraints. This yields
$\ccl{c \iseq d}{\{ x \niseq a, x \niseq b \}}$, and by applying  the $\hypcst$-Assertion rule again, we obtain the \Aclause $\ccl{\Box}{\{ x \niseq a, x \niseq b, c \niseq d \}}$, which satisfies the required property.
\end{newex}

\section{Refinements}

\label{sect:ref}

\newcommand{\proper}{{\frak P}}
\newcommand{\nice}{closed under subsumption\xspace}

Theorem \ref{theo:comp} proves that \SP-saturation permits to obtain the prime $\hypcst$-implicates of any set of clauses. This set may still be very large, it could thus require a lot of time to be generated and be difficult to handle.
In this section we introduce some refinements of the calculus \SP, showing that at almost no cost, it is possible to generate only those prime $\hypcst$-implicates of a clause set $S$ that satisfy properties that are \nice (see Definition \ref{def:nice}), or to obtain a more concise representation of all the $\hypcst$-implicates of $S$.

\subsection{Imposing Additional Restriction on the Implicates}
\label{sect:cond}
The first refinement is rather straightforward: it consists in
investigating how the calculus can be adapted to generate implicates
satisfying additional arbitrary restrictions  (e.g., for generating
implicates of some bounded cardinality, or purely positive
implicates). 
We show that some
restrictions can be imposed on the constraint part of all the
{\Aclause}s occurring in the search space without losing deductive
completeness; in other words, inferences yielding to {\Aclause}s whose
constraints do not fulfill the considered restriction can be blocked.
This is possible if these implicates belong to some class that is
closed under \MnachowithAnswer{ajout:}{ok (peut-être on pourrait mettre simplement ``subsumption'' ?} some form of logical generalization. More formally:


\begin{newdef}\label{def:nice}
A set of clauses $\proper$ is {\em \nice} if
 for every $C \in \proper$ and for every clause $D$ such that $D\sigma \subseteq C$ for some substitution $\sigma$, 
 we have $D \in \proper$.
An \Aclause $\ccl{C}{\X}$ is {\em $\proper$-compatible} if $\X^{\compl} \in\proper$.
\end{newdef}

\begin{newprop}\label{prop:compatible}
  Let $\proper$ be a set of clauses that is \nice, and let
  $\ccl{E}{\calZ}$ be an \Aclause generated by an \SP-rule, with $\ccl{C}{\X}$ as a premise. If
  $\ccl{E}{\calZ}$ is $\proper$-compatible, then so is $\ccl{C}{\X}$.
\end{newprop}

\begin{proof}
  We only consider the case where $\ccl{E}{\calZ}$ is generated by the
  $\hypcst$-Superposition rule applied to $\ccl{C}{\X}$ and
  $\ccl{D}{\Y}$, the case for the unary inference rules is
  similar. Then by definition, $\calZ = (\X \cup \Y \cup \E)\sigma$,
  where $(\sigma, \E)$ is an $(\X\cup \Y)$-pure \Asubstitution, and we
  have
    \[\X^\compl\sigma \subseteq [(\X \cup \Y\cup \E)\sigma]^\compl\ =\ \calZ^\compl.\]
  Since $\proper$ is \nice, we deduce that $\ccl{C}{\X}$ is $\proper$-compatible.
\end{proof}

\SPR{\proper} denotes the calculus \SP in which all inferences that generate non-$\proper$-compatible \Aclause are blocked.
The following theorem shows that the calculus \SPR{\proper} is deductive complete for the clauses in $\flatcl{\hypcst} \cap \proper$.

\MnachowithAnswer{modif énoncé}{oui}

\begin{newtheo}
\label{theo:compbis}
Let $S_{init}$ be a set of standard clauses and let $S$ be a set of {\Aclause}s obtained from $S_{init}$ by \SPR{\proper}-saturation.
If $\proper$ is \nice then $\abdcl{S} \dominates \imp{S_{init}} \cap \proper$.
\end{newtheo}

\begin{proof}
A simple induction together with Proposition \ref{prop:compatible} proves that all the ancestors of $\proper$-compatible clauses generated by \SP are necessarily $\proper$-compatible themselves. Since $\Box \in \proper$ for all sets $\proper$ that are \nice, all the clauses in $S_{init}$ must be $\proper$-compatible, hence the result.
\end{proof}

Examples of  classes of clauses that are \nice include the following sets that are of some practical interest:
\begin{itemize}
    \item{The set of clauses $C$ such that there exists a substitution $\sigma$ such that $C\sigma$ is equivalent to a clause of length at most $k$.}
\item{The set of positive (resp. negative) clauses.}
\item{The set of implicants of some formula $\phi$.}
\end{itemize}
Note also that the class of clause sets that are \nice is closed under union and intersection, which entails that these criteria can be combined easily.

\subsection{Discarding the Inferences on {\eflat{$\hypcst$}} Clauses}

\label{sect:ep12}

In this section we impose a restriction on the calculus that consists in preventing inferences on  $\hypcst$-literals. The obtained calculus 
is not complete since it does not generate all $\hypcst$-implicates in general, but it is complete in a restricted sense: every $\hypcst$-implicate is a logical consequence of the set of {\eflat{$\hypcst$}} clauses generated by the calculus.

\MnachowithAnswer{ajout environnement définition}{oui}
\begin{newdef}
We denote by \SPP the calculus \SP in which no inference upon $\hypcst$-literals is allowed, except for the $\hypcst$-Assertion and $\hypcst$-Reflection rules. We denote by $\projectionbis{S}{\X}$ the set obtained from $\projection{S}{\X}$ by deleting, in every clause $C \in \projection{S}{\X}$,
each literal $l$ such that the unit clause $l^{\compl}$ belongs to $\X \cup \{ \pp \iseq \true \}$.
\end{newdef}

\MnachowithAnswer{ajout}{oui. L'exemple n'était pas forcément très parlant parce que seul le littéral $\pp \niseq \trueform$ était supprimé. J'ai ajouté un littéral $a \niseq b$ dans la première clause.}
\begin{newex}
  Consider the set of clauses and \Aset from Example \ref{ex:projection}. The set $\projectionbis{S}{\X}$ contains the following clauses:
  \begin{itemize}
  \item $f(b) \iseq d$, $b\niseq d$, $d\iseq d$ and $g(t,b) \iseq
    f(d)$, where $t$ ranges over the set of all ground terms;
  \item $a \iseq b$ and $c \iseq d$;
  \item $\ff(b) \not \iseq \ff(e)$ and $\ff(d) \not \iseq \ff(e)$;
  \item $\pp \iseq \true$, $\ff(b) \niseq \ff(d) \vee b \iseq d$, $\ff(b) \niseq \ff(e)$ and $\ff(d) \niseq \ff(e) \vee d \iseq e$.
  \end{itemize}
\end{newex}

\MnachowithAnswer{ajout proposition}{oui}

\begin{newprop}\label{prop:projbis}
  For all sets of {\Aclause}s $S$ and {\Aset}s $\X$, $\projection{S}{\X} \equiv \projectionbis{S}{\X}$.
\end{newprop}

\SPP essentially simulates the calculus  in \citep{EP12a}, but there are some important differences: in particular our previous approach does not handle variable-active axioms and is complete only for implicates containing no
predicate symbol other than $\iseq$. This entails that for example, an implicate of the form $p(c_1,\dots,c_n) \iseq d$ can only be generated if a new constant
$c$ is added to  $\hypcst$, along with the axiom $c \Leftrightarrow p(c_1,\dots,c_n)$.
It is clear that applying this operation on all ground atoms is costly from a practical point of view. This is avoided with the new calculus \SPP, thanks to the addition of new inference rules.

\begin{newlem}
\label{lem:saturatedbis}
Let $S$ be an \SPP-saturated set of {\Aclause}s and let $\X$ be a complete and satisfiable \Aset.
The set $\projectionbis{S}{\X}$ is \Sup{\prec}{\selB}-saturated.
\end{newlem}

\begin{proof}
  We prove that every {\eflat{$\hypcst$}} clause of type
  \ref{projext:set} in $\projectionbis{S}{\X}$ is redundant in
  $\projectionbis{S}{\X}$. Let $C = a \bowtie b \vee C'$ be such a
  clause, by definition of $\projection{S}{\X}$, we have $a =
  \red{a}{\X}$ and $b = \red{b}{\X}$.  If $a=b$ then either $\bowtie =
  \iseq$, in which case $a \bowtie b \vee C$ is a tautology, or
  $\bowtie = \niseq$, in which case $(a\niseq a)^\compl \in \X$, and
  $C$ cannot occur in $\projectionbis{S}{\X}$. Thus $a \not = b$, and
  since $\X$ is complete, we deduce that $a \not \iseq b \in \X$ ad
  that $\ff(a) \not \iseq \ff(b)$ occurs in $\projection{S}{\X}$.
  This implies that $\bowtie$ is $\niseq$, since otherwise the literal
  $a \iseq b$ would have been deleted from the clause. Thus, $C$ is of the form $a\niseq b\vee C'$; it is
  not positive, and by construction, it contains the literal $\pp
  \niseq \true$. We deduce that $\ff(a) \niseq \ff(b) \models C$ and that
  $\ff(a) \niseq \ff(b) \prec C$; $C$ is therefore
  redundant. This implies that the only \MnachowithAnswer{non-redundant?}{oui}
  non-redundant
  inferences that can be applied on
  clauses in $\projectionbis{S}{\X}$ are upon literals that are not
  {\eflat{$\hypcst$}}.  The restriction on the calculus \SPP does not
  affect such inferences, thus, as shown in the proof of Lemma
  \ref{lem:saturated}, they can be simulated by inferences on the
  corresponding {\Aclause}s in $S$.
\end{proof}

The next theorem states a form of completeness for the restricted calculus \SPP, which is weaker than that of the calculus \SP (compare with Theorem \ref{theo:comp})  and similar to that of \citep{EP12a}.
The proof is based on the following result.

\MnachowithAnswer{modif énoncé}{En fait après m'être replongé dedans, j'ai vu qu'on avait pas besoin de $S_{init}$. Ca marche pour tout $S$ saturé. L'idée est que si $\X$ est complet, et si $\Y \not \subseteq \X$ alors forcément $\X \models \Y^c$, donc les instances des clauses de $S$ qui ne sont pas dans $\projection{S}{\X}$ sont des conséquences logiques de $\X$. }

\nikonew{ajout proposition ci-dessous}

\begin{newprop}
\label{proj:projcomplete}
Let $S$ be a set of {\Aclause}s and let $\X$ be a complete \Aset. Then
$\projectionbis{S}{\X} \models S$.
\end{newprop}
\begin{proof}
By Proposition \ref{prop:projbis} we have
$\projectionbis{S}{\X} \equiv \projection{S}{\X}$.
  Let $\ccl{C}{\Y}$ be a clause in $S$ and let $\sigma$ be a ground $\Y$-pure substitution. If $\Y\sigma \not \subseteq \X$, then there exists $l \in \Y\sigma$ such that $l \not \in \X$ and since $\X$ is complete we deduce that $l^c \in \X$, which entails that $\X \models \Y\sigma^c$, and thus $\projection{S}{\X} \models \Y\sigma^c$ (since by Proposition \ref{prop:xcont} $\projection{S}{\X} \models \X$).
  Otherwise, we must have $\red{C\sigma}{\X} \vee C' \in \projection{S}{\X}$, with $C' \subseteq \pp \niseq \trueform$ and since $\pp \iseq \trueform \in \projection{S}{\X}$ we deduce that $\projection{S}{\X} \models C\sigma$.
\end{proof}

\begin{newtheo}
\label{theo:compter}
Let $S$ be an \SPP-saturated set of {\Aclause}s. Then
$\abdcl{S} \models \imp{S}$.

\end{newtheo}
\begin{proof}
  \MnachowithAnswer{reformulation}{ok, changements pour coller au nouvel énoncé.}
  We prove the contrapositive, i.e., that every counter-model of
  $\imp{S}$ is a counter-model of $\abdcl{S}$.  Let $\M$ be a
  counter-model of $\imp{S}$ and let $\X$ be the corresponding \Aset,
  i.e.\ the set containing all {\eflat{$\hypcst$}} literals that are
  true in $\M$.  By definition, $\X$ is complete and satisfiable.
 By Proposition \ref{proj:projcomplete}, $\projectionbis{S}{\X} \models S \models \imp{S}$.
Since
  $\projectionbis{S}{\X} \models \X$ and $\X \cup \imp{S}$ is
  unsatisfiable, we deduce that $\projectionbis{S}{\X}$ is
  unsatisfiable; but $\projectionbis{S}{\X}$ is
  \Sup{\prec}{\selB}-saturated by Lemma \ref{lem:saturatedbis}, hence $\Box \in \projectionbis{S}{\X}$. We deduce that $S$
  contains an \Aclause of the form $\ccl{C}{\Y}$ and there exists a
  substitution $\sigma$ such that $\red{C^{\compl}\sigma}{\X} \cup
  \Y\sigma \subseteq \X$. Without loss of generality, we assume that
  $C$ is the  clause with the least number of literals
  satisfying this property. Assume that $C$ is nonempty. Then
  $\selB(\red{C\sigma}{\X})$ contains at least one literal $\red{(u
    \bowtie v)\sigma}{\X}$ and $C$ is of the form $u \bowtie v \vee
  D$.  If $\bowtie$ is $\iseq$, then the $\hypcst$-Assertion rule
  can be applied to this literal, yielding the \Aclause $\ccl{D}{\Y \cup \{
    u \not \iseq v \}}$. Since $S$ is \SPP-saturated, this \Aclause
  must be \redundant and by Definition \ref{def:red_crit}, $S$ contains an
  \Aclause $\ccl{D'}{\Y'}$, such that, for some substitution $\theta$, $D'\theta \subseteq D\sigma$ and
  $(\Y'\cup \set{u\niseq v})\theta \subseteq \Y\sigma$ (note that $D\sigma$ cannot be a
  tautology because $D^{\compl}\sigma \subseteq \X$ and $\X$
  is satisfiable).  This is impossible because then $\ccl{D'}{\Y'}$
  would then satisfy the above restriction, thus contradicting the
  minimality of $C$.  If $\bowtie$ is $\niseq$ then  $\red{(u
  \iseq v)\sigma}{\X}$ must occur in $\X$ since $\red{C^{\compl}\sigma}{\X} \subseteq \X$; this implies that $\red{u\sigma}{\X}
  = \red{v\sigma}{\X}$, hence that $u\sigma\almosteq{\X} v\sigma$. Thus the
  $\hypcst$-Reflection rule applies, yielding $\ccl{D}{\Y \cup
    \E}\eta$, where $(\eta,\E)$ is the most general unifier of $u$ and
  $v$.  There exists a substitution $\sigma'$ such
  that $\sigma \almosteq{\X} \eta\sigma'$, and by the same reasoning
  as previously, since $S$ is \SPP-saturated, it contains an \Aclause
  $\ccl{D'}{\Y'}$ and there exists a substitution $\theta'$ such that $D'\theta' \subseteq D\eta\sigma'$ and $\Y'\theta' \subseteq \Y\eta\sigma' \cup \E$.
  But then $(\red{D'\theta'}{\X})^\compl \subseteq (\red{D\eta\sigma'}{\X})^\compl = (\red{D\sigma}{\X})^\compl \subseteq \X$, and $\red{\Y'\theta'}{\X} \subseteq \red{(\Y\eta\sigma'\cup \E)}{\X} = \red{(\Y\sigma \cup \E)}{\X} \subseteq \X$.  Again, this contradicts the minimality of
  $C$.  Therefore, $C$ is empty, and  $(\Y\sigma)^{\compl} \in
  \abdcl{S}$. Now $\Y\sigma\subseteq \X$, thus $\M \not\models (\Y\sigma)^\compl$, which proves that  $\M$ is indeed a counter-model of $\abdcl{S}$, and the proof is completed.
\end{proof}

The difference between the calculi \SP and \SPP can be summarized as follows.
\begin{itemize}
\item{The calculus \SP explicitly generates all prime implicates in $\imp{S}$, whereas \SPP only generates a finite representation of them, in the form of an {\eflat{$\hypcst$}} implicant $S'$ of $\imp{S}$. The formula $S'$ can still contain redundancies and some additional post-processing step is required to generate explicitly the prime implicates of $S'$ if needed. Any algorithm for generating prime implicates of propositional clause sets can be used for this purpose, since flat ground equational clause sets can be reduced into equivalent sets of propositional clauses by adding equality axioms. In \citep{EPT13} a much more efficient algorithm has been proposed, in which equality axioms are directly taken into account in the inference engine and redundancy pruning mechanism. From a practical point of view, the set $\imp{S}$ can be very large, thus $S'$ can also be viewed as a concise and suitable representation of such a set.
    }
\item{The calculus \SPP restricts inferences on {\eflat{$\hypcst$}} literals to those that actually delete such literals, possibly by transferring them to the constraint part of the clauses (the $\hypcst$-Assertion and $\hypcst$-Reflection rules). From a practical point of view, this entails that these literals do not need to be considered anymore in the clausal part of the {\Aclause}: they can be transferred {\em systematically} in the constraints. This can reduce the number of generated clauses by an exponential factor, since a given {\eflat{$\hypcst$}} clause $l_1 \vee \ldots \vee l_n$ can be in principle represented by $2^n$ distinct {\Aclause}s depending on whether $l_i$ is stored to the clausal or constraint part of the \Aclause (for instance $a \iseq b$ can be represented as $\ccl{a \iseq b}{\emptyset}$ or $\ccl{\Box}{a \not \iseq b}$). Furthermore, the number of applicable inferences is also drastically reduced, since the rules usually apply in many different ways on (selected) $\hypcst$-literals, due to the fact that two {\eflat{$\hypcst$}} terms are always \Aunifiable and that the ordering $\alwayssgreater{}$ is empty when applied on terms in $\hypcst$. For example the clauses
    $a \iseq b$ and $c \iseq d$ generate the {\Aclause}s
    \[\ccl{d \iseq b}{\{ a \iseq c \}},\
    \ccl{d \iseq a}{\{ b \iseq c \}},\
    \ccl{c \iseq b}{\{ a \iseq d \}},\
    \ccl{c \iseq a}{\{ b \iseq d \}},\] regardless of the ordering $\prec$.
    }
\end{itemize}

The following example illustrates the differences between \SP and \SPP.

\begin{newex}
Let $S = \{ f(a,b) \niseq f(c,d), g(x) \iseq 0 \vee x \iseq c, g(a) \niseq 0 \}$, where $x \in \V$, $g(x) \succ a \succ b \succ c \succ d$ and $\hypcst = \{ a,b,c,d \}$.
It is easy to check that
\SPP generates the implicates $\ccl{\Box}{ \{ a \iseq c, b \iseq d \}}$ (by $\hypcst$-Reflection on the first clause) and $\ccl{\Box}{\{ a \niseq c \}}$ (by an application of the $\hypcst$-Superposition rule from the second clause into the third one, followed by an application of the $\hypcst$-Assertion rule).
However, the implicate $\ccl{\Box}{\{ b \iseq d \}}$ that is a logical consequence of the above {\Aclause}s is not generated.
In contrast, it is possible to infer this implicate with \SP: First the $\hypcst$-Superposition rule generates as usual the clauses $a \iseq c$ and then $f(c,b) \niseq f(c,d)$ (the constraints are empty at this point since all the considered {\Aunifier}s are standard unifiers), and $\ccl{\Box}{\{ b \iseq d \}}$ is inferred by applying $\hypcst$-Reflection on the latter clause.
Note that \SP has a larger search space than \SPP. Consider for instance a clause $a \iseq b \vee c \iseq d$. \SPP simply reduces this clause into
$\ccl{\Box}{ \{ a \niseq b, c \niseq d \}}$ and no further inference is applicable on it, while
\SP also generates the {\Aclause}s
$\ccl{a \iseq b}{\{ c \niseq d \}}$ and
$\ccl{c \iseq d}{\{ a \niseq b \}}$, which in turn possibly enable other inferences.
\end{newex}
It is possible to combine the two calculi \SP and \SPP. This can be done as follows.
\begin{itemize}
\item{Starting from a set of clauses $S$, \SPP is first applied until saturation, yielding a new set $S'$. By Theorem \ref{theo:compter} we have $\abdcl{S'} \equiv \imp{S}$.}
\item{Then \SPR{\proper} is applied on $\abdcl{S'}$ until saturation yielding a set $S''$, where $\proper$ denotes the set of clauses that logically entail at least one clause in $\abdcl{S'}$. It is clear that this set of clauses is \nice, hence by Theorem \ref{theo:compbis}, we eventually obtain a set of clauses $\abdcl{S''} \dominates \imp{\abdcl{S'}} \cap \proper$.
    But $\imp{\abdcl{S'}}\cap \proper \dominates \abdcl{S'}$, hence $\abdcl{S''} \dominates \abdcl{S'}$, and $\abdcl{S''} \equiv \imp{S}$. The set of clauses $\abdcl{S''}$ can therefore be considered as a concise representation of $\imp{S}$. This approach is appealing since $\abdcl{S''}$ is in generally much smaller than $\imp{S}$, and contrary to $\abdcl{S'}$, this set is free of redundancies.}
\end{itemize}
Another straightforward method to eliminate redundant literals from the clauses in $\abdcl{S'}$  without having to explicitly compute the set $\imp{S'}$ is to test, for every clause $l \vee C \in \abdcl{S'}$, whether the relation $\abdcl{S'} \models C$, holds, in which case the literal $l$ can be safely removed. The test  can be performed by using any decision procedure for ground equational logic \citep[see for instance][for a similar approach]{Meir:2005:YDP:2153230.2153271,DBLP:conf/sas/DilligDA10}.
Note however that removing redundant literals is not sufficient to obtain prime implicates, as shown in the following example.
\begin{newex}
\label{ex:paar}
Consider the clause set: $S \isdef \{ a \not \iseq c \vee b \not \iseq c \vee d \iseq e,\ a \iseq c \vee a \iseq f,\ b \iseq c \vee a \iseq f,\ f \not \iseq b \}$. It is easy to check that $a \niseq b \vee d \iseq e$ is an implicate of $S$ and that this clause is strictly more general than $a \not \iseq c \vee b \not \iseq c \vee d \iseq e$. The calculus \SP computes the \Aclause
$\ccl{\Box}{ \{a \iseq b, d \iseq e \}}$, yielding the
set of prime implicates:
$S' \isdef \{ a \not \iseq b \vee d \iseq e, a \iseq c \vee a \iseq f, b \iseq c \vee a \iseq f, f \not \iseq b \}$. $S'$ is equivalent to $S$ and strictly smaller. In contrast, the approach devised by \cite{DBLP:conf/sas/DilligDA10} cannot simplify $S$ since there is no useless literal.
\end{newex}

\section{Termination}

\newcommand{\stronglysgreater}{\vartriangleright_{\hypcst}}
\newcommand{\stronglygreater}{\trianglerighteq_{\hypcst}}
\newcommand{\stronglyssmaller}{\vartriangleleft_{\hypcst}}
\newcommand{\stronglysmaller}{\trianglelefteq_{\hypcst}}

\label{sect:term}

\newcommand{\projectionter}[2]{\Gamma(#1,#2)}
\newcommand{\ords}{\preccurlyeq}
\newcommand{\strsel}{\sel_{\hypcst}}

\newcommand{\stronglyredundant}{strongly redundant\xspace}

We relate the termination behavior of \SP to that of the usual Superposition calculus.
We first introduce restricted ordering and redundancy criteria.
For all expressions (terms, atoms, literals or clauses) $t$ and $s$, we write
$t \stronglysgreater s$ if $t' \succ s'$ holds for all expressions $t',s'$ such that
$t \equivto t'$ and $s \equivto s'$. Note that the ordering $\stronglysgreater$ is stronger
than $\alwaysgreater{}$ (and also stronger than $\succeq$) because the constants \MnachowithAnswer{constants?}{oui} in $t$ and $s$ can be rewritten independently of each other.
Assume for instance that $\prec$ is such that $f(a) \prec g(a) \prec f(b) \prec g(b)$ with $\hypcst = \{ a,b \}$.
Then it is easy to check that $g(a) \alwaysgreater{} f(a)$ but $g(a) \not \stronglysgreater f(a)$ since $g(a) \prec f(b) \equivto f(a)$.
Also, let $\strsel$ be the selection function defined from the function $\sel$ as follows: for every clause $l \vee C$,
$l \in \strsel(l \vee C)$ if there exists $l',C'$ such that $l' \equivto l$, $C' \equivto C$ and
$l' \in \sel(l' \vee C')$.
We show that most termination results for the calculus
\Sup{\stronglyssmaller}{\strsel} also apply to \SP. To this purpose, we consider a restricted form of redundancy testing.
\begin{newdef}

 A standard clause $C$ is {\em \stronglyredundant} in a set of standard clauses $S$ iff for every clause $C' \equivto C$, $C'$ is \redundant in $S$.
\end{newdef}

\begin{newdef}
For every set of {\Aclause}s $S$ and for every ground \Aset $\Y$,
we denote by $\projectionter{S}{\Y}$ the set of standard clauses
 $C\sigma$, where $\ccl{C}{\X} \in S$ and $\sigma$ is  an $\X$-pure substitution of
 domain $\var(\X)$ such that $\red{\sigma}{\Y} =\sigma$ and $\X\sigma \subseteq \Y$.
\end{newdef}

The definition of $\projectionter{S}{\Y}$ is similar to that of
$\projection{S}{\Y}$ (see Section \ref{sect:comp}), except that: (i) only the variables occurring in $\X$ are instantiated; (ii) the clauses are not reduced with respect to the equations in the constraint part (but the constants replacing the variables in $\X$ are reduced).
\begin{newex}
Let $S = \{ \ccl{f(x,y) \iseq a}{\{ x \niseq b \} } \}$ with $\hypcst = \{ a,b,c \}$ and $a \succ b \succ c$.
We have $\projectionter{S}{\{ a \niseq b, a \iseq c \}} = \{ f(c,y) \iseq a \}$ and
$\projectionter{S}{\{ a \niseq b, c \niseq b \}} = \{ f(a,y) \iseq a, f(c,y) \iseq a \}$.
\end{newex}

 \begin{newlem}
\label{lem:term}
Let $S$ be a set of {\Aclause}s, $E$ be an \Aclause and $\U$ be a ground \Aset.
\begin{itemize}
\item{If $E$ can be deduced from $S$ by \SP, then every clause  in $\projectionter{E}{\U}$
can be deduced from $\projectionter{S}{\U} \cup \U$
by \Sup{\stronglyssmaller{}}{\strsel}.}
\item{If $S$ is a set of standard clauses and $\projectionter{E}{\U}$  contains a clause that is \stronglyredundant in $\projectionter{S}{\U}$ then $E$ is \redundant in $S$.}
\end{itemize}
\end{newlem}

\newcommand{\fstsub}{\gamma}
\newcommand{\smgu}{\nu}

\begin{proof}
See Appendix \ref{ap:term}.
\end{proof}

We denote by $\U_{\hypcst}$ the set of all unit clauses of the form $p(a_1,\dots,a_n) \bowtie \true$ or $a \bowtie b$, with $a_1,\dots,a_n,a,b \in \hypcst$.
For any set of clauses $S$, we denote by $S^{\star}$ the set of clauses inductively defined as follows.
\begin{itemize}
\item{$S \subseteq S^{\star}$.}
\item{If $C$ is not \stronglyredundant in $S$ and is deducible from $S^{\star} \cup \U_{\hypcst}$ by applying the rules in \Sup{\stronglyssmaller}{\strsel} (in one step), then $C \in S^{\star}$.}
\end{itemize}
Lemma \ref{lem:term} immediately entails the following:
\begin{newcor}
\label{cor:term}
Let $S$ be a set of clauses.
If $S^\star$ is finite then \SP terminates on $S$ (up to redundancy).
\end{newcor}
In order to prove that \SP terminates on some class of clause sets ${\frak S}$, it suffices to prove that $S^{\star}$ is finite, for every $S \in {\frak S}$.
The calculus \Sup{\stronglyssmaller{}}{\sel} is slightly less restrictive than the usual Superposition calculus \Sup{\prec}{\sel}, since $\stronglyssmaller{}$ is a stronger relation than $\prec$.
However, most of the usual termination results for the Superposition calculus still hold
for \Sup{\stronglyssmaller{}}{\sel}, because they are closed under the addition of equalities between constants
and do not depend on the order of $\almosteq{}$-equivalent terms.
Similarly, redundancy testing is usually restricted to subsumption and tautology detection.
In particular, all the termination results described by \citet*{ARR03} are preserved
(it is easy to check that $S^{\star}$ is finite for the considered sets of axioms).

An interesting continuation of the present work would be to devise  formal (automated) proofs of the termination of \SP on the usual theories of interest in program verification, enriched by arbitrary ground clauses.
This could be
done by using
existing schematic calculi \citep[see, e.g.,][]{LynchM-LICS02,Lynch20111026,DBLP:conf/rta/TushkanovaRGK13} to compute a symbolic
representation of the set of {\Aclause}s $S^\star$.

\section{Conclusion and Discussion}


Although the Superposition calculus is not deductive-complete in general, we have shown
that it can be adapted in order to make it able to generate all implicates defined over a given {\em finite} set of {\em ground} terms denoted by constant symbols, using a finite set of predicate symbols including the equality predicate. Furthermore, this is done in such a way that the usual termination properties of the calculus are preserved. By duality, the procedure can be used to generate abductive explanations of first-order formul{\ae}.

A major restriction of our approach is that it cannot handle built-in theories such as arithmetics which play an essential rôle in verification. Axiomatizing these theories in first-order logic is infeasible or inefficient. A natural follow-up of this work is therefore to make the procedure  able to cooperate with external decision procedures. This can be done for instance by combining our approach with existing techniques for fusing the Superposition calculus and external reasoning tools \citep{bachmair1994refutational,AlthausKW09,DBLP:conf/cade/BaumgartnerW13}. These techniques, based on the use of constrained Superposition together with an abstraction of the terms of the considered theory, should be easy to combine with $\hypcst$-Superposition.
Note that our calculus has many commun points with the above-mentioned constrained Superposition calculi, however in our case the constraint and clausal parts are not defined over disjoint signatures: in contrast the \Aunification and Assertion rules allow one to transfer literals from the clausal part to the constraints. In other approaches \citep{bachmair1994refutational,AlthausKW09,DBLP:conf/cade/BaumgartnerW13} the constraints are used to store formul{\ae} that cannot be handled by the Superposition calculus, whereas in our case they are used to store properties that are {\em asserted} instead of being proved.

Another obvious drawback with the calculi \SP and \SPP is that the user has to explicitly declare the set of abducible terms (i.e., the constants in $\hypcst$). This set must be finite and must contain all \MnachowithAnswer{all?}{oui} built-in constants. Note that, thanks to the results in Section \ref{sect:ref}, unsatisfiable or irrelevant implicates (such as
$0 \iseq 1$) can be easily detected and discarded on the fly during proof search.
Handling infinite (but recursive) sets of terms is possible from a theoretical point of view: it suffices to add an inference rule generating clauses of the form $a \iseq t$, where $t$ is an abducible ground terms and $a$ is a fresh abducible constant symbol. It is easy to see that completeness is preserved, but of course termination is lost. A way to recover termination  is to develop additional techniques to restrict the application of this rule by selecting the terms $t$. This could be done either statically, from the initial set of clauses, or dynamically, from the information deduced during proof search.

 Another possible extension would be to generate ``mixed'' implicates, containing both abducible and non-abducible terms, which would avoid having to declare built-in constants as abducible.
An alternative approach consists in avoiding to have to explicitly declare abducible terms, by
adding rules for generating them symbolically (as the $\hypcst$-Substitutivity rule does for predicate symbols). For termination, additional conditions should be added to ensure that the set of abducible terms is finite (using, e.g., sort constraints).


Another restriction is that our method does not handle non-ground abducible terms, hence cannot generate quantified formul{\ae}. We are now investigating these issues.


\bibliography{biblio,Nicolas.Peltier}
\bibliographystyle{spbasic}

\appendix

\nikonew{j'ai mis l'appendix en normalsize à voir par la suite}

{\normalsize

\section{$\hypcst$-Unification}

\label{ap:unif}

\begin{newdef}
An {\em \Aunification problem} is either $\bot$  or a triple $(S,\theta,\X)$ where $S$ is a  set of equations, $\theta$ is a substitution such that $x\theta = x$ for every variable $x$ occurring in $S$ and $\X$ is a positive \Aset.
A pair $(\sigma,\X)$ is a \emph{solution} of an \Aunification problem $\P = (S,\theta,\X')$ iff
$(\theta,\X') \moregeneral (\sigma,\X)$ and $(\sigma,\X)$ is an \Aunifier of every $t \iseq s \in S$.
An \Aunification problem is {\em satisfiable} if it has a solution.
\end{newdef}
The set of {\em \Aunification rules} is the set of rules depicted in Figure \ref{unifrules}. They are almost identical to the standard unification rules, except that equations of the form $a \iseq b$ where $a \not = b$ do not lead to failure but are instead stored in $\X$.
We assume that the rules are applied in the specified order, i.e., a rule applies only if the previous rules do not apply. Note that, following our convention, $\X \cup \{ a \iseq b \}$ actually denotes the smallest \Aset containing $\X$ and $a \iseq b$ (obtained by transitive closure from $\X \cup \{ a \iseq b \}$).
\begin{figure}
{\small \begin{tabular}{llll}
(T) & $(t \iseq t \cup S,\theta,\X)$ & $\rightarrow$ & $(S,\theta,\X)$ \\
(E) & $(a \iseq b \cup S, \theta,\X)$ & $\rightarrow$ &
$(S,\theta,\X \cup \{ a \iseq b \})$ \, if $a,b\in \hypcst$ \\
(C) & $(\{ f(t_1,\ldots,t_n) \iseq g(s_1,\ldots,s_m) \} \cup S, \theta,\X)$ & $\rightarrow$ &
$\bot$ \, if $f \not = g$ \\
(O) & $(\{ x \iseq t[x]_p \} \cup S,\theta,\X)$ & $\rightarrow$ & $\bot$  \\
(R) & $(\{ x \iseq t \} \cup S,\theta,\X)$ & $\rightarrow$ & $(S\{x \mapsto t \}, \theta \cup \{ x \mapsto t \},\X)$ \\
(D) & $(\{ f(t_1,\ldots,t_n) \iseq f(s_1,\ldots,s_n) \} \cup S, \theta,\X)$ & $\rightarrow$ &
$(\bigcup_{i=1}^n \{ t_i \iseq s_i \} \cup S,\theta,\X)$ \\
\end{tabular}}
\caption{\Aunification rules\label{unifrules}}
\end{figure}
\begin{newlem}
The  \Aunification rules preserve the set of solutions of the considered problem.
\end{newlem}
\begin{proof}
The proof is by an easy inspection of each rule (see Figure \ref{unifrules} for the notations):
\begin{itemize}
\item[(T)]{We have $t\sigma \almosteq{\X'} t\sigma$, for all $t,\sigma,\X'$; hence removing the equation $t \iseq t$ does not affect the set of solutions.}
\item[(E)]{Since (T) is not applicable, $a$ and $b$ are distinct. Thus
    we have $a\sigma \almosteq{\X'} b\sigma$ iff $a \iseq b \in
    \X'$. Consequently, adding the equation $a \iseq b$ to the last
    component of the problem does not affect the set of
    solutions. Furthermore, every \Asubstitution $(\sigma,\X')$ such that
    $\X \cup \{ a \iseq b \} \subseteq \X'$ is an \Aunifier of $a$ and
    $b$, thus the equation can be removed from the first component of
    the problem once it has been added to the last component.}
\item[(C)]{Due to the ordering of the rules, $f$ and $g$ must be
    distinct and cannot both occur in $\hypcst$, since otherwise, one
    of $(T)$ or $(E)$ would apply first. Thus, the problem has no
    solution, since by definition of the relation $\almosteq{\X'}$, we
    have $f(t_1,\ldots,t_n)\sigma \not \almosteq{\X'}
    g(s_1,\ldots,s_m)\sigma$, for all $\sigma,\X'$.}
\item[(O)]{Due to the ordering of the rules, $p$ cannot be empty,
    since otherwise $(T)$ would apply first, thus $x\sigma$ contains
    strictly less positions than $t\sigma$. Hence $x\sigma
    \not \almosteq{\X} t\sigma$, for all $\sigma,\X$, and the problem
    has no solution.}

\item[(R)]{If an \Asubstitution $(\sigma,\X')$ is a unifier of $x$ and $t$ then necessarily
$x\sigma \almosteq{\X'} t\sigma$. Thus adding the mapping $x \mapsto t$ to the substitution and replacing $x$ by $t$ does not affect the set of solutions.
Afterwards the equation $x \iseq t$ becomes trivial and can be removed.
}

\item[(D)]{If is clear that $f(t_1,\ldots,t_n)\sigma \almosteq{\X} f(s_1,\ldots,s_n)\sigma$ holds
iff for all $i \in [1,n]$, $t_i\sigma \almosteq{\X} s_i\sigma$ holds. Thus the replacement of the equation $f(t_1,\ldots,t_n) \iseq f(s_1,\ldots,s_n)$  by the set $\{ t_i \iseq s_i \mid i \in [1,n] \}$ preserves the set of solutions.}
    \end{itemize}
\end{proof}

\begin{newcor}
Every satisfiable \Aunification problem has a most general \Aunifier, which is unique up to $\equivto$-equivalence.
\end{newcor}

\begin{proof}
It is easy to check that the \Aunification rules terminate: all the rules strictly decrease the size of the first component of the problem, except for (R), which strictly decreases the number of variables occurring in the first component (moreover, no rule can increase this number of variables).
Furthermore, irreducible problems are either $\bot$ or of the form $(\emptyset,\sigma,\X)$. In the former case the problem has no solution and in the latter, $(\sigma,\X)$ is a most general solution.
Also, if $(\sigma,\X)$ and $(\sigma',\X')$ are two most general solutions then by definition we have
$(\sigma,\X) \moregeneral (\sigma',\X')$ and
$(\sigma',\X') \moregeneral (\sigma,\X)$, thus $(\sigma,\X) \equivto (\sigma',\X')$.
\end{proof}
Note that the proposed algorithm is exponential w.r.t.\ the size of the initial problem, however it can be easily transformed into a polynomial algorithm by using structure sharing (thus avoiding any duplication of terms).

\section{Proof of Lemma \ref{lem:saturated}}

\label{ap:saturated}

The proof is based on the following intermediate results.

\begin{newdef}
  Let $S$ be a set of {\Aclause}s and $\X$ be an \Aset.  If $u \bowtie
  v \vee C \vee D$ is a clause of type \ref{projext:set} in
  $\projection{S}{\X}$, where $\bowtie \in \set{\iseq, \niseq}$ and $D
  \subseteq \{ \pp \niseq \true \}$, then there exist an
  \Aclause $\ccl{u' \bowtie v' \vee C'}{\Y}\in S$ and a substitution
  $\sigma$ such that\footnote{If several terms $u'$ satisfying the above
    conditions exist then one of them is chosen arbitrarily.} $\red{(u'\sigma)}{\X} = u$, $\red{(v'\sigma)}{\X}=v$,
  $\red{(C'\sigma)}{\X} = C$ and $\Y\sigma \subseteq \X$.  The term occurrence
  $u$ is {\em \fixed} in $u \bowtie v \vee C$ if  $u'$ occurs in $\Y$ whenever it is a
  variable.
\end{newdef}

 \begin{newprop}
 \label{prop:fixedbis}
 Let $S$ be a set of {\Aclause}s and $\X$ be an \Aset. Let $C$ be a
 clause of type \ref{projext:set} in $\projection{S}{\X}$ and $a,b$ be  constants in $\hypcst$ such that
 $\red{a}{\X} = a$ and $\red{b}{\X}=b$.
Let $P$ be a set of non-\fixed occurrences of $a$ in $C$. Then there exists a set $P'$ of occurrences of $a$ in $C$ that contains $P$, and a clause $D$ in
$\projection{S}{\X}$ such that $D$  is obtained from $C$ by replacing all occurrences of $a$ in $P'$ by $b$.
 \end{newprop}

\begin{proof}
  By definition, there exists an \Aclause $\ccl{C'}{\Y}\in S$ and a
  substitution $\sigma$ such that $C = \red{C'\sigma}{\X} \vee C''$,
  $C'' \subseteq \set{\pp \niseq \true}$ and $\Y\sigma \subseteq \X$.  Since
  $P$ is a set of non-\fixed occurrences in $C$, the subterms of $D$
  at the positions in $P$ are variables $x_1,\ldots,x_n$ not occurring
  in $\Y$.

  Consider the substitution $\theta$ coinciding with $\sigma$, except
  that $\forall i \in [1,n]$, $x_i\theta \isdef b$.  Since the
  variables $x_i$ ($1 \leq i \leq n$) do not occur in $\Y$, $\theta$
  and $\sigma$ coincide on $\Y$, hence $\Y\theta \subseteq \X$. This
  means that $\projection{S}{\X}$ must contain the clause of type
  \ref{projext:set} $\red{C'\theta}{\X} \vee C''$ (note that $C''$ is
  not affected because $C'\sigma$ is {\eflat{$\hypcst$}} and positive
  exactly when $C'\theta$ satisfies the same property).  By
  definition, since $\red{a}{\X} = a$ and $\red{b}{\X} = b$,
  $\red{C'\theta}{\X} \vee C''$ is therefore obtained from $C =
  \red{C'\sigma}{\X} \vee C''$ by replacing some occurrences of $a$ by
  $b$, and in particular, all the occurrences in $P$ are replaced.
\end{proof}

Note that $P'$ can be a strict superset of $P$: for example, if $S =
\{ x \iseq c \vee x \iseq d \}$, then $a \iseq c \vee a \iseq d
\in \projection{S}{\emptyset}$,  position $1.1$ is not \fixed in
$a \iseq c \vee a \iseq d$, and it is clear that $b \iseq c \vee b \iseq d  \in\projection{S}{\emptyset}$ but $b \iseq c \vee a \iseq c  \not \in \projection{S}{\emptyset}$.


\MnachowithAnswer{pourquoi ne pas décaler cette proposition à l'annexe?}{done}

\begin{newprop}
\label{prop:postaut}
Let $\ccl{C}{\X}$ be an {\Aclause}; assume that $\X$ is satisfiable and that  $C$ is \eflat{$\hypcst$} \MnachowithAnswer{\eflat{$\hypcst$}?}{oui} and \quasipositive.
Then $\ccl{C}{\X}$ is a tautology if and only if  $\red{C}{\X}$ is either a tautology or contains a literal that also occurs in $\X$.
In particular, if $C$ is \elementary and positive then
$\ccl{C}{\X}$ is a tautology  exactly when
$\red{C}{\X}$ is a tautology.
\end{newprop}
\Mnacho{reprise}
\begin{proof}
  Assume that $\red{C}{\X}$ is not a tautology and contains no literal
  in $\X$. Let $I$ be the interpretation such that $\forall a,b \in
  \hypcst$, $I \models a\iseq b$ iff $\red{a}{\X} = \red{b}{\X}$ and
  for all $a_1,\dots,a_n$ where $\red{a_i}{\X} = a_i$, $I \models
  p(a_1,\dots,a_n) \iseq \true$ iff $p(a_1,\dots,a_n) \iseq \true \in
  \X$ or $p(a_1,\dots,a_n) \niseq \true \in \red{C}{\X}$. Note that
  $I$ is well-defined, since $\X$ and $\red{C}{\X}$ share no literals
  and neither of them contains complementary
  literals. 
  By definition, $I$ validates all positive literals in $\X$. If $a
  \not \iseq b \in \X$ and $I \not \models a \not \iseq b$, then
  $\red{a}{\X} = \red{b}{\X}$, hence $\X \models a \iseq b$, which
  means that $\X$ is unsatisfiable, and this contradicts the
  hypothesis of the lemma.  Similarly, if $p(a_1,\dots,a_n) \niseq
  \true \in \X$ and $I \models p(a_1,\dots,a_n) \iseq \true$ then
  since $\X$ is satisfiable, $p(a_1,\dots,a_n) \niseq \true$ must
  occur in $\red{C}{\X}$, which contradicts the hypothesis that $\X$
  and $\red{C}{\X}$ share no literals. Therefore, $I \models \X$. Now
  consider a literal $l \in C$. Since $C$ is {\eflat{$\hypcst$}} and
  \quasipositive, $l$ is of the form $a \iseq b$ or $p(a_1,\dots,a_n)
  \bowtie \true$.  If $l$ is of the form $a \iseq b$ and $\red{a}{\X}
  = \red{b}{\X}$ then $\red{C}{\X}$ is a tautology, and this is
  impossible by hypothesis. Thus $\red{a}{\X} \not = \red{b}{\X}$ and
  $I \not \models a \iseq b$.  Now assume that $l$ is of the form
  $p(a_1,\dots,a_n) \iseq \true$ and that $I \models l$; the case
  where $l$ is of the form $p(a_1,\dots,a_n) \niseq \true$ is
  similar. Let $m = p(\red{a_1}{\X},\dots,\red{a_n}{\X}) \iseq
  \true$. Since $I\models l, \X$, it is clear that $I\models m$, thus by
  definition of $I$, either $m \in \X$ or $m^\compl \in
  \red{C}{\X}$. In the first case $m$ occurs in both $\X$ and
  $\red{C}{\X}$, and in the second case, both $m$ and $m^\compl$ occur
  in $\red{C}{\X}$ which is a tautology; thus we get a contradiction
  in both cases.  Therefore, $I$ is a counter-model of
  $\ccl{C}{\X}$.

The converse is straightforward.
\end{proof}
Note that the previous property does not hold if $C$ is not \nikonew{modif} \quasipositive;
for example, $\ccl{a \niseq b}{a \niseq b}$ is a tautology but the unit clause
$\red{(a \niseq b)}{\{ a \niseq b \}} = a \niseq b$ is not.

\begin{newlem}
\label{lem:composeq}
Let $S$ be an \SP-saturated set of {\Aclause}s and $\X$ be an \Aset.
For $i = 1,2$, let $u_i \iseq v_i \vee C_i$
be an {\eflat{$\hypcst$}} clause of type \ref{projext:set} in $\projection{S}{\X}$, and assume that $u_i \not = v_i$.
If the following conditions hold:
\begin{itemize}
\item $u_1 = u_2 \not = \true$,
\item $u_1$ is \fixed in  $u_1 \iseq v_1 \vee C_1$,
\item for $i = 1, 2$, $u_i\iseq v_i \in \selB(u_i \iseq v_i \vee C_i)$,
\item  $v_1 \iseq v_2 \vee C_1 \vee C_2$ is not a tautology,
\end{itemize}
then $\projection{S}{\X}$ contains a clause of type \ref{projext:set}
contained in $v_1 \iseq v_2 \vee C_1 \vee C_2$.
\end{newlem}

\begin{proof}
For $i = 1, 2$, since $u_i \iseq v_i \vee C_i$ is of type \ref{projext:set}, 
there exists an {\Aclause}
$\ccl{t_i \iseq s_i \vee D_i}{\Y_i}\in S$ and a substitution $\sigma_i$ such that
 $\Y_i\sigma_i \subseteq \X$,
$\red{D_i\sigma_i}{\X} = C_i$, $\red{t_i\sigma_i}{\X} = u_i$ and $\red{s_i\sigma_i}{\X} = v_i$.

Let $\sigma = \sigma_1\sigma_2$.
Since $u_1=u_2$, we have $t_1\sigma = t_1\sigma_1 \almosteq{\X} t_2\sigma_2 = t_2\sigma$, hence $(\sigma, \X)$ is an \Aunifier of $t_1 \iseq t_2$.
Let $(\eta, \calZ)$ be a most general \Aunifier of $t_1\iseq t_2$, then $\calZ \subseteq \X$, and there exists a ground substitution $\sigma'$ such that $\forall x$, $x\eta\sigma' \almosteq{\X} x\sigma$.
Now, $\red{[(t_i\iseq s_i)\eta\sigma']}{\X} = u_i\iseq v_i$, which is selected in $u_i\iseq v_i\vee C_i$, and
since $\sel$ is stable under {\Asubstitution}s,
$(t_i \iseq s_i)\eta \in \sel((t_i \iseq s_i \vee D_i)\eta)$.

By hypothesis $v_1,v_2 \in \hypcst \cup \V$, hence $s_1,s_2 \in
\hypcst \cup \V$. By definition of $\alwayssgreater{\Y_i}$, this
implies that $s_i\eta \not \alwaysgreater{\Y_i\eta} t_i\eta$: indeed,
$s_1,s_2$ can be replaced by the minimal constant $\true$, either by
instantiation or by rewriting of constants in $\hypcst$. Note also that
$s_i\eta\neq t_i\eta$ since otherwise we would have $u_i=v_i$, which contradicts the
hypotheses of the lemma.

Since $u_1$ is \fixed, either $t_1$ is not a variable or $t_1$ occurs
in $\Y_1$, hence by definition of \SP, the $\hypcst$-Superposition
 from $\ccl{t_2
  \iseq s_2 \vee D_2}{\Y_2}$ into $\ccl{t_1 \iseq s_1 \vee D_1}{\Y_1}$  upon the terms $t_1$ and $t_2$ generates
$\ccl{(s_1 \iseq s_2\vee D_1 \vee D_2)\eta}{\Y_1\eta \cup \Y_2\eta
  \cup \calZ}$.  Now, the \Aclause $\ccl{(s_1 \iseq s_2\vee D_1 \vee
  D_2)\eta\sigma'}{\Y_1\eta\sigma' \cup \Y_2\eta\sigma' \cup \calZ}$
must be \redundant in $S$, because $S$ is \SP-saturated. This clause cannot be a tautology; indeed, for $i
= 1, 2$, since $u_i \iseq v_i \in \selB(u_i \iseq v_i \vee C_i)$ and
$u_i \iseq v_i \vee C_i \in\flatcl{\hypcst}$, $C_i$ must be positive
by definition of the selection function $\selB$,  and cannot contain a symbol in $\preds$ (otherwise the literal containing this symbol would be strictly greater than $u_i \iseq v_i$).
By hypothesis, $v_1
\iseq v_2 \vee C_1 \vee C_2$ is not a tautology and since $C_1,C_2$
are positive and \elementary, we deduce by Proposition \ref{prop:postaut} that
$\ccl{(s_1 \iseq s_2\vee D_1 \vee D_2)\eta\sigma'}{\Y_1\eta\sigma'
  \cup \Y_2\eta\sigma' \cup \calZ}$ is not a tautology either. Thus,
by Definition \ref{def:red_crit}, there exists an \Aclause
$\ccl{E}{\calZ'} \in S$ and a substitution $\theta$ such that $E\theta
\subseteq (s_1 \iseq s_2\vee D_1 \vee D_2)\eta\sigma'$
and $\calZ'\theta \subseteq \Y_1\eta\sigma' \cup
\Y_2\eta\sigma' \cup \calZ \subseteq \X$.
Therefore, $\projection{S}{\X}$ contains the clause $\red{E\theta}{\X}$ that
is  contained in $\red{[(s_1 \iseq s_2\vee D_1 \vee D_2)\eta\sigma']}{\X}= v_1 \iseq v_2 \vee C_1 \vee C_2$.
 \end{proof}

\begin{newprop}
\label{prop:simulatefact}
Let $S$ be an \SP-saturated set of {\Aclause}s and $\X$ be a ground \Aset.
If $\projection{S}{\X}$ contains a non-tautological clause $D \subseteq C \vee a \iseq b \vee a \iseq b$, where $C \vee a \iseq b$  is positive and \elementary then $C \vee a \iseq b$ is redundant in $\projection{S}{\X}$.
\end{newprop}

\begin{proof}
  If $D$ contains at most one occurrence of $a \iseq b$, then
  necessarily $D \subseteq C \vee a \iseq b$ and the proof is
  immediate. Otherwise, since $a\iseq b\vee a\iseq b \subseteq D$, the latter cannot be of type 2; it is therefore of
  type $1$, thus there exists an \Aclause $\ccl{D'}{\Y} \in S$ and a
  substitution $\sigma$ such that $\red{D'\sigma}{\X} = D$ and
  $\Y\sigma \subseteq \X$.  $D'$ is of the form $C' \vee u \iseq v
  \vee u' \iseq v'$, where $\red{C'\sigma}{\X} \subseteq C$,
  $\red{u\sigma}{\X} = \red{u'\sigma}{\X} = a$ and $\red{v\sigma}{\X}
  = \red{v'\sigma}{\X} = b$.  By Proposition \ref{prop:flatred}, $D'$
  is \eflat{$\hypcst$}; thus the literal $u \iseq v$ is necessarily
  $\alwaysgreater{}$-maximal in $D'$, and the $\hypcst$-Factorization
  rule applied to $\ccl{D'}{\Y}$ generates $\ccl{(C' \vee u \iseq v
    \vee v \not \iseq v')\theta}{\Y\theta \cup \calZ}$, where
  $(\theta,\calZ)$ is the m.g.u.\ of $u$ and $u'$, or
  simply $\ccl{(C' \vee u \iseq v)\theta}{\Y\theta \cup \calZ}$, if
  $v\theta = v'\theta$. We assume that $v\theta \neq v'\theta$, the
  proof when they are equal is simpler. Since $(\sigma,\X)$ is an
  instance of $(\theta,\calZ)$, by Proposition \ref{prop:instred} the
  clause $\ccl{(C' \vee u \iseq v \vee v \not \iseq v')\sigma}{\X}$
  must be redundant in $S$, and since
  $\red{v\sigma}{\X} = \red{v'\sigma}{\X} = b$, it is
  equivalent to $\ccl{(C' \vee u \iseq v)\sigma}{\X}$.  This \Aclause
  cannot be a tautology; otherwise, by Proposition \ref{prop:postaut},
  $\red{(C' \vee u \iseq v)\sigma}{\X} \equiv D$ would also be a
  tautology. Therefore, by Definition \ref{def:red_crit}, there exists an
  \Aclause $\ccl{E}{\U} \in S$ and a substitution $\eta$ such that
  $E\eta \subseteq (C' \vee u \iseq v \vee v \not \iseq v')\sigma$ and
  $\U\eta \subseteq \X$.  By definition of $\projection{S}{\X}$, the
  clause $\red{E\eta}{\X}$ occurs in $\projection{S}{\X}$.  If $E\eta
  \subseteq (C' \vee u \iseq v)\sigma$ then the proof is completed.
  Otherwise, $E$ is of the form $E' \vee w \not \iseq w'$, where
  $E'\eta \subseteq (C' \vee u \iseq v)\sigma$, $w\eta = v\sigma$ and
  $w'\eta = v'\sigma$.  Note that $w$ and $w'$ cannot both be equal to
  $\true$, since otherwise $w \niseq w'$ would have been removed from
  the \Aclause, thus the literal $w \not \iseq w'$ is necessarily
  $\alwaysgreater{}$-maximal in $E' \vee w \not \iseq w'$, and it must
  be selected; therefore, the $\hypcst$-Reflection rule can be applied on this
  clause. Since $(\eta,\X)$ is a unifier of $w$ and $w'$, necessarily,
  the \Aclause $\ccl{E'\eta}{\X}$ is redundant in $S$.  By Definition \ref{def:red_crit}, $S$ contains a clause $\ccl{E''}{\U'}$ and
  there exists a substitution $\mu$ such that $E''\mu \subseteq
  E'\eta$ and $\U'\mu \subseteq \X$. We conclude that
  $\red{E'\eta}{\X} \subseteq C \vee a \iseq b$ must be redundant in
  $\projection{S}{\X}$.
\end{proof}


\begin{newprop}\label{prop:xcont}
  There exists a set of clauses $U\subseteq
  \projectiontype{S}{\X}{\ref{projext:eq}} \cup
  \projectiontype{S}{\X}{\ref{projext:noeq}} \cup
  \projectiontype{S}{\X}{\ref{projext:pred}} \cup
  \projectiontype{S}{\X}{\ref{projext:pp}}$ such that $U$ contains no
  occurrence of $\pp$, and $\U \models \X$ 
\end{newprop}

\begin{proof}
  Consider the following sets:
  \begin{eqnarray*}
    X_1 &\isdef &\setof{a \iseq \red{a}{\X}}{a\in \X,\, a\neq
      \red{a}{\X}},\\
    X_2 &\isdef &\setof{\red{(a\niseq
        b)}{\X}}{a\niseq b \in \X}, \\
    X_3 &\isdef &\setof{\red{(f(a_1,\dots,a_n)\bowtie
        \true)}{\X}}{f(a_1,\dots,a_n)\bowtie
        \true \in \X}.
  \end{eqnarray*}
  It is clear that $\X \equiv X_1\cup X_2 \cup X_3$ and that $X_1 \subseteq
  \projectiontype{S}{\X}{\ref{projext:eq}}$. By letting $X_2'\isdef
  \bigcup_{a\niseq b \in X_2}\set{\ff(a) \niseq \ff(b)}$, we have $X_2' \models X_2$, \nikonew{non pas $X_2 \equiv X_2'$} the set $U\isdef
  X_1\cup X_2' \cup X_3$ entails $\X$, it is a subset of
  $\projectiontype{S}{\X}{\ref{projext:eq}} \cup
  \projectiontype{S}{\X}{\ref{projext:noeq}} \cup
  \projectiontype{S}{\X}{\ref{projext:pred}}$
and contains no occurrence
  of $\pp$.
\end{proof}

We now establish a result concerning the form of the clauses of type \ref{projext:noeq} or \ref{projext:pred}
in $\projection{S}{\X}$.

\nikonew{Lemme commun pour les types \ref{projext:noeq} et \ref{projext:pred}.}

\begin{newlem}
\label{lem:projextform}
Any clause $C$ of type \ref{projext:noeq} (resp. \ref{projext:pred}) in $\projection{S}{\X}$ is of the form
$\ff(a_1) \niseq \ff(a_2) \vee C'$ (resp. $p(a_1,\dots,a_n) \bowtie \trueform \vee C'$) where:
\begin{enumerate}
\item{$a_1,a_2 \in \hypcst$ (resp. $a_1,\dots,a_n \in \hypcst$) \nikonew{cet item est redondant, mais peut être plus clair comme ça, à voir} }
\item{$C'$ is positive and \elementary.}
\item{$\X$ contains a clause of the form $b_1 \niseq b_2$ (resp. $p(b_1,\dots,b_n)$).}
\item{For every $i \in [1,2]$ (resp. $i \in [1,n]$) either $a_i = b_i$ or $a_i \prec b_i$ and $\projection{S}{\X}$ contains a clause of the form
$a_i \iseq b_i \vee C_i$ with $C_i \subseteq C'$ and $C_i \prec (a_i \iseq b_i)$.\label{item:key}}
\end{enumerate}
\end{newlem}
\begin{proof}
By definition of the clauses of type \ref{projext:noeq} and \ref{projext:pred}  in $\projection{S}{\X}$ (see Definition \ref{projext:def}), $C$ is obtained from a clause of the form  $\ff(b_1) \niseq \ff(b_2)$  with $b_1 \niseq b_2 \in \X$ (resp. from a clause $p(b_1,\dots,b_n) \in \X$)
by applying Superposition inferences from positive \elementary clauses in $\projection{S}{\X}$. Furthermore, $C$ cannot be redundant. We prove the result
by induction on the number of Superposition inferences.
The base case is immediate (with $a_i = b_i, C' = \Box$).
Assume that $C$ is obtained by Superposition into a clause $D$. Without loss of generality we assume that the considered derivation is minimal (w.r.t. the number of steps).
By the induction hypothesis, $D$ is necessarily of the form $\ff(a_1) \niseq \ff(a_2) \vee C'$ (resp.
 $p(a_1,\dots,a_n) \bowtie \trueform \vee C'$), where $a_1,\ldots,a_n$ and $C'$ satisfy the above properties. By definition of the selection function $\selB$, only the literal $\ff(a_1) \niseq \ff(a_2)$ (resp. $p(a_1,\dots,a_n)$) is selected, hence the replacement necessarily occurs in this literal.
  By symmetry, we may assume that it occurs upon the constant $a_1$, from a clause of the form $a_1 \iseq a_1' \vee D'$ (with $a_1 \succ a_1'$). The inference yields
 $C = \ff(a_1') \niseq \ff(a_2) \vee C' \vee D'$ (resp. $p(a_1',a_2,\dots,a_n) \bowtie \trueform \vee C' \vee D'$).
If $b_1 = a_1$, then the proof is completed, since the clause $a_1 \iseq a_1' \vee D'$  fulfills the property of Item \ref{item:key}. Otherwise, by the induction hypothesis,
$\projection{S}{\X}$ contains a clause of the form
$a_1 \iseq b_1 \vee C_1$ with $C_1 \subseteq C'$
Assume that $a_1$ is not \fixed in $a_1 \iseq b_1 \vee C_1$.
By Proposition \ref{prop:fixedbis}, this entails that
$\projection{S}{\X}$ contains a clause of type \ref{projext:set} of the form
$a_1' \iseq b_1 \vee C_1'$, where $C_1'$ is obtained from $C_1$ by replacing occurrences
of $a_1$ by $a_1'$ ($b_1$ is not replaced, since $b_1 \not = a_1$).
By replacing the Superposition inference upon $b_1$ in the derivation yielding $C$ by a Superposition from $a_1' \iseq b_1 \vee C_1'$, we get a clause $D''$ of the form
$a_1' \niseq a_2 \vee C''$ (resp. $p(a_1',a_2,\dots,a_n) \bowtie \trueform \vee C''$ with $C'' \subseteq C' \vee C_1'$.
Clause $D''$ satisfies the following properties.
\begin{itemize}
\item{$D''$ is a clause of type \ref{projext:noeq} or \ref{projext:pred} in $\projection{S}{\X}$.}
\item{$D'' \preceq C$, since $a_1' \prec a_1$.}
\item{$a_1' \iseq a_1 \vee D' \prec C$, since by definition of the ordering $\ff(x) \succ c$ and $p(\vec{x}) \succ c$ for every $c \in \hypcst$.}
\item{$D'', a_1' \iseq a_1 \vee D' \models C$.}
\end{itemize}
The number of inferences in the derivation is strictly lower than that of $C$ (since the sequence of Superposition inferences replacing $b_1$ by $a_1$ and then $a_1$ by $a_1'$ has been  replaced by a single replacement of $b_1$ by $a_1'$), which by minimality of the derivation entails that $D'' \not = C$. Thus $D'' \succ C$ and
$C$ is redundant, which contradicts the definition of the clauses of type \ref{projext:noeq} and
\ref{projext:pred}.
Consequently $a_1$ is \fixed in $a_1 \iseq b_1 \vee C_1$. We now distinguish two cases.
\begin{itemize}
\item{
The clause
    $b_1 \iseq a_1' \vee C_1 \vee D'$ is a tautology. Since this clause is positive, this entails that it contains a literal of the form $t \iseq t$ (otherwise the interpretation mapping all constants to distinct elements would falsify the clause).
    Since $b_1 \succeq a_1$ and $a_1 \succ a_1'$ we have $b_1 \not = a_1'$ hence the literal $t \iseq t$ occurs in $C_1 \vee D'$. But then
    $C$ would be redundant (since it contains $C_1 \vee D'$), which contradicts the definition of the clauses of type \ref{projext:noeq} and
    \ref{projext:pred}.
}
\item{
    The clause $b_1 \iseq a_1' \vee C_1 \vee D'$ is not a tautology.
    Since
    $a_1$ is \fixed in $a_1 \iseq b_1 \vee C_1$, by Lemma \ref{lem:composeq}, we deduce that
    there is a clause of type \ref{projext:set} in
    $\projection{S}{\X}$ that is contained in $(a_1' \iseq b_1) \vee
    C_1 \vee D'$. If this clause is contained in $C_1 \vee
    D'$ then it is also contained in $C$ which is redundant and the
    proof is completed; otherwise it is of the form $(a_1' \iseq b_1)
    \vee C_1'$, where $C_1' \subseteq C_1 \vee D' \subseteq C' \vee D'$, which proves
    that the above property holds for $C$.}
    \end{itemize}
\end{proof}

 We are now in a position to provide the proof of Lemma \ref{lem:saturated}.
\label{proof:saturated}
We have to prove that every clause generated from $\projection{S}{\X}$ by an inference in \Sup{\prec}{\selB} except for \Eqfact Factorization on positive {\eflat{$\hypcst$}} clauses is a logical consequence of some clauses in $\projection{S}{\X}$ that are strictly smaller than the maximal premise of the inference.  Note that this condition necessarily holds if
the conclusion is redundant in $\projection{S}{\X}$, since a clause cannot be greater than its maximal premise.
We distinguish several cases, depending on the types of the clauses involved in the inference.

\MnachowithAnswer{changé paragraph en subsubsection parce que l'espacement n'était pas top avec paragraph}{oui}

\subsubsection*{Clauses of type \ref{projext:eq}.}

By definition, every such clause is of the form $c \iseq \red{c}{\X}$,
where $c\neq \red{c}{\X}$ and by construction, $c \succ \red{c}{\X}$. Constant $c$
cannot occur in other clauses in $\projection{S}{\X}$, since all its
occurrences are replaced by $\red{c}{\X}$. Thus the clause $c
\iseq \red{c}{\X}$ cannot interact with any other clause, because of the
ordering restrictions of the Superposition calculus.

\subsubsection*{Clauses of type \ref{projext:pp}.}
By construction, constant $\pp$ only occurs in literals of the form $\pp\niseq \true$ and $\pp \iseq \true$. By definition of $\selB$, the literal $\pp \not \iseq \true$ is never selected, thus
the clause $\pp \iseq \true$ cannot interact with other clauses in $\projection{S}{\X}$.
Now, consider a clause of the form $\ff(u) \niseq \ff(v) \vee u \iseq v$. By definition, $u = \red{u}{\X}$, and $u$ cannot be the maximal term of a selected literal in $\projection{S}{\X}$. Since $\ff$ occurs only in negative literals, no literal can interact with
$\ff(u) \niseq \ff(v)$, and since $u \not = v$, the Reflection rule does not apply either.

\subsubsection*{Clauses of type \ref{projext:noeq}.}
Let $C$ be a clause of type \ref{projext:noeq}. By definition, only
negative literals are selected in $C$, thus the only inference rules
that can be applied on $C$ are the Reflection rule or the
Superposition rule into $C$, where the ``from'' premise is necessarily a clause of
type \ref{projext:set} in $\projection{S}{\X}$.  By Case
\ref{projext:noeq} of Definition \ref{projext:def},
 all the non-redundant clauses that can be generated by the Superposition inference rule are already in $\projectiontype{S}{\X}{\ref{projext:noeq}}$. Thus, we only consider the case where the Reflection inference rule applied on $C$ generates a clause $D$.

\nikonew{note: pas besoin de $\red{a_1}{\X}$, car toutes les constantes de $\projection{S}{\X}$ sont supposées en forme normale.}

By Lemma \ref{lem:projextform},
$C$ is of the form $\ff(a_1) \niseq \ff(a_2) \vee C'$, where
$\X$ contains a clause of the form $b_1 \niseq b_2$ with for all $i =1,2$ either $b_i = a_i$ or
$\projection{S}{\X}$
 contains a clause of type
    \ref{projext:set} of the form $(b_i \iseq a_i) \vee C_i$, where $C_i \subseteq C'$.
    Furthermore, by definition of the Reflection rule, we must have $a_1 = a_2$.

If $b_1 = a_1$ or $b_2 = a_2$ or
if $a_i$ is \fixed in $(b_i \iseq a_i) \vee F_i$, then by Lemma \ref{lem:composeq}, $\projection{S}{\X}$ contains a clause $(b_1 \iseq b_2) \vee C''$ with $C'' \subseteq C_1 \vee C_2$.  Then $S$ contains an \Aclause of the form
    $\ccl{u \iseq v \vee E}{\Y}$, where $\red{u\theta}{\X} = b$, $\red{v\theta}{\X} =
    a$, $\red{E\theta}{\X} = C''$ and $\Y\theta \subseteq \X$.
    Then
    the $\hypcst$-Assertion rule can be applied to this \Aclause,
    yielding $\ccl{E}{\Y \cup \{ u \not \iseq v \}}$.
    Note that
    since $b_1 \neq b_2$, $\Y\theta \cup
    \set{u \niseq v}\theta$ must be satisfiable.
    If $\ccl{E}{\Y \cup \{
      u \not \iseq v \}}\theta$ is a tautology, then so is
    $\red{E\theta}{\X} = C''$ by Proposition \ref{prop:postaut}, hence
    $D$ is also a tautology and is redundant in $\projection{S}{\X}$,
    thus the proof is completed. Otherwise, by Definition \ref{def:red_crit},
    since $S$ is \SP-saturated, it contains an
    \Aclause $\ccl{E'}{\Y'}$ and there exists a substitution $\theta'$
    such that ${E'\theta'} \subseteq E\theta$ and $\Y'\theta' \subseteq \Y\theta
    \cup \set{u \niseq v}\theta$.
    Then $\projection{S}{\X}$ contains the clause
    $\red{E'\theta'\theta}{\X} \subseteq \red{E\theta}{\X} = C''
    \subseteq C'$, and the latter is therefore redundant in
    $\projection{S}{\X}$.

Now assume that $b_1 \not = a_1$, $b_2 \not = a_2$ and that neither $a_1$ nor $a_2$ is \fixed.
By Proposition \ref{prop:fixedbis}, $\projection{S}{\X}$ contains a clause
of the form $b_1 \iseq b_2 \vee G_1$, where $G_1$ is obtained from $F_1$ by replacing occurrences of
$a_1$
by $b_2$.
Using the fact that $S$ is saturated under $\hypcst$-Assertion, we deduce as in the previous case that
$\projection{S}{\X}$ contains a clause $G_1' \subseteq G_1$.
Thus, since $\ff(a_1) \niseq \ff(a_2) \vee F_1 \vee F_2 \subseteq C$ and $a_1 = a_2$, we have:
\[G_1', b_2 \iseq a_2 \vee F_2\ \models\ G_1, b_2 \iseq a_2 \vee F_2\
\models\ G_1\vee F_2\ \models\ C.\] Since $C$ contains an occurrence
of $\ff$, it is strictly greater than $G_1'$ and $b_2 \iseq a_2 \vee
F_2$, thus $C$ is redundant, and cannot be a
clause of type \ref{projext:noeq}.

\subsubsection*{Clauses of type \ref{projext:pred}.}

\nikonew{ajout}

By Lemma \ref{lem:projextform}, $C$ is necessarily of the form
$p(a_1,\dots,a_n) \bowtie \true \vee C'$,
where for every $i \in [1,n]$, one of the two following conditions hold:
 \begin{enumerate}
 \item{$a_i = b_i$.}
 \item{$\projection{S}{\X}$ contains a positive \elementary clause of the form
$a_i \iseq b_i \vee C_i$, with $a_i \prec b_i$, $(a_i \iseq b_i) \succ C_i$, $C_i \subseteq C'$.}
\end{enumerate}
The only rule that can be applied on $C$ (beside Superposition from \elementary positive clauses for which the proof follows immediately from Case \ref{projext:pred} of Definition \ref{projext:def}) is the Superposition rule on the term $p(a_1,\dots,a_n)$, and in this case the other premisse must be of the form
$p(a_1,\dots,a_n) \not \bowtie \true \vee F$. The generated clause is $C' \vee F$, since literals of the form $\true \niseq \true$ are deleted.

 By definition of $\projection{S}{\X}$, for each index $i$ satisfying the second item, there exist an \Aclause
$\ccl{a_i' \iseq b_i' \vee C_i'}{\Y_i} \in S$ and a substitution $\sigma_i$ such that \MnachowithAnswer{pas de réduction?}{si \smiley}
$\red{a_i'\sigma_i}{\X} = a_i$, $\red{b_i'\sigma_i}{\X} = b_i$, $\red{C_i'\sigma_i}{\X} = C_i$, and $\Y_i\sigma \subseteq \X$.
Let $E$ (resp. $E'$) be the disjunction of the clauses $C_i$ (resp. $C_i'$), for all indices such that $a_i \not = b_i$. Note that $E \subseteq C'$, hence it is sufficient to prove that $E \vee F$ is redundant in $\projection{S}{\X}$.
The $\hypcst$-Substitutivity rule applied on  the clauses $\ccl{a_i' \iseq b_i' \vee C_i'}{\Y_i}$ generates the \Aclause:
$\ccl{p(a'_1,\dots,a'_n) \bowtie \true \vee E'}{\{ p(b'_1,\dots,b'_n) \bowtie \true \}}$,
with $b_i = a_i \Rightarrow (b'_i = a'_i = x_i)$ (where the $x_i$'s denote pairwise distinct fresh variables)
and $b_i \not = a_i \Rightarrow (b'_i = b_i \wedge a'_i = a_i)$.
 This \Aclause must be redundant in $S$, in particular (taking $x_i = b_i$ if $b_i = a_i$) either
 $\ccl{p(a_1,\dots,a_n) \bowtie \true \vee E}{\{ p(b_1,\dots,b_n) \bowtie \true \}}$ is a tautology (Case (i)) or
 there exist an \Aclause
 $\ccl{D}{\Y}$ and a substitution $\theta$ with
 $\red{D\theta}{\X} \subseteq p(a_1,\dots,a_n) \bowtie \true \vee E$
 and
 $\Y\theta \subseteq \{ p(b_1,\dots,b_n) \bowtie \true \}$ (Case (ii)).

 If $\red{D\theta}{\X} \subseteq E$ then $E \vee F$ is clearly redundant in $\projection{S}{\X}$, thus
 we assume that $D$ is of the form $p(\vec{s}) \bowtie \true \vee D'$, with
 $\red{\vec{s}\theta}{\X} = (a_1,\dots,a_n)$ and
 $\red{D'\theta}{\X} \subseteq E$. Note that by definition of the ordering $p(\red{\vec{s}\theta}{\X}) \bowtie \true$ is strictly greater than any literal in $\red{D'\theta}{\X}$.

By Proposition \ref{prop:postaut}, we observe that $\ccl{p(a_1,\dots,a_n) \bowtie \true \vee E}{\{ p(b_1,\dots,b_n) \bowtie \true \}}$ is a tautology only if $(a_1,\dots,a_n) = (b_1,\dots,b_n)$.
We then distinguish two cases, according to the type of
the other premisse $p(a_1,\dots,a_n) \not \bowtie \true \vee F$.

\begin{enumerate}
\item{
If  $p(a_1,\dots,a_n) \not \bowtie \true \vee F$ is of type \ref{projext:set}, then there exist an \Aclause $\ccl{p(\vec{t}) \not \bowtie \true \vee F'}{\calZ} \in S$ and a substitution $\theta'$ such that
$\red{\vec{t}\theta'}{\X} = (a_1,\dots,a_n)$, $\red{F'\theta'}{\X} \vee F'' = F$ and $\calZ\theta' \subseteq \X$, where $F'' = \pp \niseq \true$ if $F$ is \quasipositive and $F'' = \Box$ otherwise.
Then:
\begin{itemize}
\item{In Case (i), we have $(a_1,\dots,a_n) = (b_1,\dots,b_n)$, by the above remark. Furthermore, the $\hypcst$-Assertion rule applies on $\ccl{p(\vec{t}) \not \bowtie \true \vee F'}{\calZ}$, yielding
$\ccl{F'}{\calZ \cup \{ p(\vec{t}) \bowtie \true \}}$. Since $S$ is saturated under the $\hypcst$-Assertion rule, this \Aclause is redundant in $S$. Since $(p(\red{\vec{t}\theta'}{\X}) \iseq \true) = (p(a_1,\dots,a_n) \iseq \true)  = (p(b_1,\dots,b_n) \iseq \true)\in \X$, this entails that $\red{F'\theta'}{\X} \vee F''$ (hence also $E \vee F$ since $F = \red{F'\theta'}{\X} \vee F''$) is redundant in $\projection{S}{\X}$.
}
\item{In Case (ii), since $\red{\vec{t}\theta'}{\X} = (a_1,\dots,a_n) = \red{\vec{s}\theta}{\X}$, $\vec{t}$ and $\vec{s}$ have an \Aunifier $(\mu,\U)$, that is more general than $(\theta \cup \theta',\X)$. Furthermore,
    $(p(\vec{s}) \bowtie \true)\mu$ and $(p(\vec{t}) \not \bowtie)\mu$ must be selected in $D\mu$ and $(p(\vec{t}) \not \bowtie \true \vee F')\mu$, respectively, because  the selection function is stable under \Asubstitution and $p(\red{\vec{s}\theta}{\X}) \bowtie \true$ and $p(\red{\vec{t}\theta'}{\X}) \not \bowtie$ must be selected in $\red{D\theta}{\X}$ and $\red{(p(\vec{s}) \not \bowtie \true \vee F')\theta'}{\X}$ respectively. Consequently, the Superposition rule applies on $\ccl{p(\vec{s}) \bowtie \true \vee D'}{\Y}$ and $\ccl{p(\vec{t}) \not \bowtie \true \vee F'}{\calZ}$, yielding $\ccl{(D' \vee F')}{\Y \cup \calZ \cup \U}\mu$. The \Aclause
    $\ccl{D'\theta \vee F'\theta'}{\X}$ is thus redundant in $S$, hence $E\vee F$ is redundant in $\projection{S}{\X}$.}
\end{itemize}}
\item{
Otherwise, $p(a_1,\dots,a_n) \not \bowtie \true \vee F$ must be of type \ref{projext:pred}, $F$ must be positive and \elementary, and by the same reasoning as before we can prove that $\X$ contains a clause $p(b_1',\dots,b_n') \not \bowtie \true$, such that
either
 $\ccl{p(a_1,\dots,a_n) \not \bowtie \true \vee F}{\{ p(b_1',\dots,b_n') \not \bowtie \true \}}$ is a tautology (Case (iii)) or
 there exist an \Aclause
 $\ccl{p(\vec{s}') \not \bowtie \true \vee D''}{\Y'}$ and a substitution $\theta'$ with
 $\red{\vec{s}'\theta'}{\X} = (a_1,\dots,a_n)$, $\red{D''}{\X}\theta \subseteq F$
 and
 $\Y'\theta \subseteq \{ p(b_1',\dots,b_n') \bowtie \true \}$ (Case (iv)).

By Proposition \ref{prop:postaut}, Case (iii) can only occur if $(a_1,\dots,a_n) = (b_1',\dots,b_n')$
Also, we note that Cases (i) and (iii) cannot hold simultaneously (otherwise we would have $(b_1,\dots,b_n) = (a_1,\dots,a_n) = (b_1',\dots,b_n')$ hence $\X$ would contain both $p(a_1,\dots,a_n) \bowtie \true$ and $p(a_1,\dots,a_n) \not \bowtie \true$ and would be thus unsatisfiable). By symmetry, we may assume that (i) does not hold. Then:
\begin{itemize}
\item{In Case (iii), we can apply the $\hypcst$-Assertion rule  on $\ccl{p(\vec{s}) \bowtie \true \vee D'}{\Y}$, yielding $\ccl{D'}{\Y  \cup \{ p(\vec{s}) \not \bowtie \true \}}$.  Since $(a_1,\dots,a_n) = (b_1',\dots,b_n')$, we have $\Y\theta \cup \{ (p(\red{\vec{s}}{\X}) \not \bowtie \true)\theta \} \subseteq \{ p(b_1,\dots,b_n) \bowtie \true,  p(b_1',\dots,b'_n) \not \bowtie \true \} \subseteq \X$ and $D'\theta \subseteq E$ is thus redundant in $\projection{S}{\X}$.
    }
\item{In Case (iv), it is easy to check that
we can apply the $\hypcst$-Superposition rule on $\ccl{p(\vec{s}) \bowtie \true \vee D'}{\Y}$ and $\ccl{p(\vec{s}') \not \bowtie \true \vee D''}{\Y'}$, yielding an \Aclause of the form $\ccl{D' \vee D''}{\Y \cup \Y' \cup \U}\mu$, where $(\mu,\U)$ is more general than $(\theta\cup \theta',\X)$.
Then $E \vee F$ is redundant in $\projection{S}{\X}$.
}
\end{itemize}
}
\end{enumerate}

\subsubsection*{Clauses of type \ref{projext:set}.}
All inferences involving a clause of type \ref{projext:eq},
\ref{projext:noeq} or \ref{projext:pred} have already been considered, we now focus on
inferences involving only clauses of type \ref{projext:set}.  We assume the Superposition rule is applied; the proof for the unary inference
rules is similar.  Let $C = u \iseq v \vee D$ and $E = t \bowtie s
\vee F$ be two clauses of type \ref{projext:set} in
$\projection{S}{\X}$.  Assume that the Superposition rule applies from
$C$ into $E$, upon the terms $u$ and $t|_{p}$, yielding $t[v]_p
\bowtie s \vee F \vee D$, where $t|_p = u$, $u \succ v$, $t \succ s$,
$u \iseq v \in \selB(C)$ and $t \bowtie s \in \selB(E)$.
Note that
this implies that $u \iseq v$ is strictly maximal in $C$.  We
prove that the clause $t[v]_p \bowtie s \vee F \vee D$ is redundant in
$\projection{S}{\X}$.
Note that by definition of $\selB$, $t \bowtie s$ cannot be $\pp \niseq \true$.
By definition, $S$ contains two {\Aclause}s $C'
= \ccl{u' \iseq v' \vee D'}{\Y}$ and $E' = \ccl{t' \bowtie s' \vee
  F'}{\calZ}$ and there exist substitutions $\sigma$ and $\theta$
such that:
\begin{itemize}
\item  $\red{u'\sigma}{\X} = u$, $\red{v'\sigma}{\X} = v$,
$\red{D'\sigma}{\X} \vee D'' = D$, $\Y\sigma \subseteq \X$ and $D''
\subseteq \set{\pp \niseq \true}$,
\item $\red{t'\theta}{\X} = t$, $\red{s'\theta}{\X} = s$,
$\red{F'\theta}{\X} \vee F'' =F$, $\calZ\theta \subseteq \X$ and $F'' \subseteq \set{\pp \niseq \true}$.
\end{itemize}
First assume that there is a strict prefix $q$ of $p$ such that $t'|_q$ is a variable $x$. Then $x$ cannot occur in $\calZ$, since otherwise $x\theta$ would be a constant in $\hypcst$ (because $\calZ\theta \subseteq \X$), 
and $q$ would not be a strict prefix of $p$.
Let
$\theta'$ be the substitution coinciding with $\theta$, except for the value of $x$, and such that $x\theta'$ is obtained from $x\theta$ by replacing all occurrences of $u$ by $v$.
Since $\theta$ and $\theta'$ coincide on all the variables in $\calZ$, necessarily
$\calZ\theta' \subseteq \X$.
Furthermore, since $(t' \bowtie s' \vee F')\theta'$ is \eflat{$\hypcst$} and positive exactly when $(t' \bowtie s' \vee F')\theta$ is \eflat{$\hypcst$} and positive, we deduce that $\red{(t' \bowtie s' \vee F')\theta'}{\X} \vee F'' \in \projection{S}{\X}$, and this clause is such that
\begin{eqnarray*}
  \red{(t' \bowtie s' \vee F')\theta'}{\X} \vee F'', u\iseq v\vee D& \models&
  \red{(t' \bowtie s' \vee F')\theta}{\X} \vee F''\vee D,\ u\iseq v\vee D\\
 & = & t\bowtie s\vee F\vee D, u\iseq v\vee D\\
 & \models & t[v]_p \bowtie s \vee F \vee D.
\end{eqnarray*}
If $\red{(t' \bowtie s' \vee F')\theta'}{\X} \vee F'' = t[v]_p \bowtie
s \vee F \vee D$ then $t[v]_p \bowtie s \vee F \vee D$ occurs in
$\projection{S}{\X}$ hence the proof is completed. Otherwise $\red{(t'
  \bowtie s' \vee F')\theta'}{\X} \vee F'' \prec t[u]_p \bowtie s \vee
F$.  If $p \neq \emptypos$ or $\bowtie = \niseq$, then necessarily
$u\iseq v \prec t[u]_p\bowtie s$, since $u \succ v$. Furthermore,
$D\prec u\iseq v$, hence $\red{(t' \bowtie s' \vee
  F')\theta'}{\X} \vee F'', u \iseq v \vee D \prec t[u]_p \bowtie s
\vee F$, and the clause $t[v]_p \bowtie s \vee F \vee D$ is therefore
a logical consequence of clauses of $\projection{S}{\X}$
that are strictly smaller than one of
its premises, the proof is thus completed.

If $p = \emptypos$ and $\bowtie = \iseq$, then $E = u\iseq s \vee F$,
and the generated clause is $v \iseq s \vee D\vee F$. If $v=s$ then
this clause is a tautology, and is trivially redundant in
$\projection{S}{\X}$. Otherwise, assume w.l.o.g.\ that $v\prec s$ (since the same inference can be performed by considering $E$ as the ``from'' premise, the two parent clauses play symmetric rôles), then
$u\iseq v \prec u\iseq s$, and as in the previous case, the clause
$v\iseq s \vee F \vee D$ is therefore a logical consequence of clauses
that are strictly smaller than one of its premises. 

Now assume that there is no strict prefix $q$ of $p$ such that $t'|_q$ is a variable $x$. Necessarily, $p$ must be a position in $t'$.
Since $u = t|_p$, we have $u'\sigma \almosteq{\X} t'|_{p}\theta$, hence $u'$ and $t'|_{p}$ are \Aunifiable.
Let $(\eta,\X')$ be a most general \Aunifier of $u'$ and $t'|_{p}$.
Since $(\sigma\theta,\X)$ is an \Aunifier of $u'$ and $t'|_{p}$ we have $\X' \subseteq \X$ and there exists a substitution $\eta'$ such that $\eta\eta'  \almosteq{\X} \sigma\theta$.
Since $\red{u'\sigma}{\X} = u \succ v = \red{v'\sigma}{\X}$,
we have $v'\eta\not \alwaysgreater{} u'\eta$, and similarly,
$t'\eta \not \alwaysgreater{} s'\eta$.
Furthermore, since the selection function $\sel$
is stable by \Asubstitution,
$(t' \bowtie s')\eta$ and
$(u' \iseq v')\eta$ must be selected in $C'\eta$ and $E'\eta$ respectively.
Thus the $\hypcst$-Superposition rule applied to $C'$ and $E'$, generates
$\ccl{(t'[v']_p \bowtie s' \vee D' \vee F')\eta}{(\Y \cup \calZ)\eta \cup \X'}$.
Since $S$ is \SP-saturated, 
this clause is \redundant in $S$,  and so
is $\ccl{(t'[v']_p \bowtie s' \vee D' \vee F')\eta\eta'}{(\Y \cup
  \calZ)\eta\eta' \cup \X'}$.

Suppose that $(t'[v']_p \bowtie s' \vee D' \vee F')\eta\eta'$ is
{\eflat{$\hypcst$}} and \quasipositive.  \nikonew{modifs}
If $\ccl{(t'[v']_p \bowtie s' \vee D' \vee F')\eta\eta'}{(\Y \cup
  \calZ)\eta\eta' \cup \X'}$ is a tautology, then by Proposition \ref{prop:postaut}
 $\red{(t'[v']_p \bowtie s' \vee D' \vee F')\eta\eta'}{\X}$ either is a tautology or
 contains a literal $A \bowtie \true$ occurring in $\X$. In both cases, $\red{(t'[v']_p \bowtie s' \vee D' \vee F')\eta\eta'}{\X}$ is redundant in $\projection{S}{\X}$
Otherwise,  by Definition \ref{def:red_crit},
$S$ contains an \Aclause $\ccl{G}{\U}$  and there exists
a substitution $\mu$ such that $G\mu \subseteq (t'[v']_p \bowtie s'
\vee D' \vee F')\eta\eta'$ and $U\mu\subseteq (\Y \cup
  \calZ)\eta\eta' \cup \X' \subseteq \X$. The clause
$G\mu$ must be positive and {\eflat{$\hypcst$}}, hence by Case
\ref{projext:set} of Definition \ref{projext:def},
$\projection{S}{\X}$ contains $\red{G\mu}{\X} \vee \Box = G\mu$, and $G\mu \subseteq \red{(t'[v']_p \bowtie s' \vee D' \vee F')\eta\eta'}{\X} = t[v]_p \bowtie s \vee D\vee F$.

If $(t'[v']_p \bowtie s' \vee D' \vee F')\eta\eta'$ is not {\eflat{$\hypcst$}} or not \nikonew{modifs} \quasipositive,
then there exist $n$ {\Aclause}s
$\ccl{C_1}{\X_1},\ldots,\ccl{C_n}{\X_n}$ and substitutions $\gamma_1,\ldots,\gamma_n$ such that:
\begin{itemize}
\item $\forall i \in [1,n]\, \X_i\gamma_i \subseteq  (\Y \cup \calZ)\eta\eta' \cup \X'$,
\item $\X'', C_1\gamma_1,\ldots,C_n\gamma_n \models (t'[v']_p \bowtie s' \vee D' \vee F')\eta\eta'$,
\item $\red{(t'[v']_p \bowtie s' \vee D' \vee F')\eta\eta'}{\X}
  \alwaysgreater{\X} C_1\gamma_1,\ldots,C_n\gamma_n$.
\end{itemize}
Since $\X'' \subseteq \X$, we deduce that
$\X, C_1\gamma_1,\ldots,C_n\gamma_n \models (t'[v']_p \bowtie s' \vee D' \vee F')\eta\eta'$. Also, by definition of $\alwaysgreater{}$, we have
$\red{C_1\gamma_1}{\X},\ldots,\red{C_n\gamma_n}{\X} \preceq \red{(t'[v']_p \bowtie s' \vee D' \vee F')\eta\eta'}{\X}$.
But since $\X_1\gamma,\ldots,\X_n\gamma \subseteq (\Y \cup \calZ)\eta\eta' \cup \X' \subseteq \X$, $\projection{S}{\X}$ contains  clauses of the form
$\red{C_i\gamma_i}{\X} \vee G_i$ ($1 \leq i \leq n$), where $G_i \subseteq \set{\pp \not \iseq  \true}$.
By Proposition \ref{prop:xcont}, $\X$ is a logical consequence of \nikonew{not equivalent to}  a subset of $\projectiontype{S}{\X}{\ref{projext:eq}} \cup
  \projectiontype{S}{\X}{\ref{projext:noeq}} \cup
 \projectiontype{S}{\X}{\ref{projext:pred}}  \cup  \projectiontype{S}{\X}{\ref{projext:pp}}$ that contains no occurrence of $\pp$.
  Since $(t'[v']_p \bowtie s' \vee D' \vee F')\eta\eta'$ is either not
  {\eflat{$\hypcst$}} or not positive, $t[v]_p \bowtie s \vee F \vee
  D$ contains $\pp \niseq \true$, and must be strictly greater than
  all clauses of type \ref{projext:eq}, \ref{projext:noeq}, \ref{projext:pred} or
  \ref{projext:pp} that do not contain any occurrence of $\pp$.  Thus
  $\red{(t'[v']_p \bowtie s' \vee D' \vee F')\eta\eta'}{\X} = t[v]_p
  \bowtie s \vee F \vee D$ is redundant in
  $\projection{S}{\X}$.\label{proof:saturatedend}

\section{Proof of Corollary \ref{cor:unsat}}

\label{ap:unsat}

Let $S'$ be the smallest set of (standard ground) clauses such that $S'$ contains all
clauses $C$ satisfying the following properties:
\begin{itemize}
\item{$C$ is generated by  one of the rules in \Sup{\prec}{\sel_{\Phi}} from $\projection{S}{\X} \cup S'$.}
\item{$C$ is not a logical consequence of the set of clauses in $\projection{S}{\X} \cup S'$ that are strictly smaller than the maximal premise of $C$.}
\end{itemize}
Let $S'' = \projection{S}{\X} \cup S'$. Intuitively, $S''$ is the \SP-closure of $\projection{S}{\X}$ modulo redundancy. 
By definition $S''$ must be unsatisfiable and weakly \Sup{\prec}{\sel}-saturated, hence $S''$ contains
the empty clause.
For any term $t$, we denote by $\spc{t}$ the set of positive
clauses in $\projection{S}{\X}$ that contain no term $s \succeq t$.
We prove that  the 
clauses \nikonew{J'ai supprimé "non-redundant". Si on met "non-redundant", il faudrait dire précisément redondant par rapport à quoi. Si c'est par rapport à $S'$ il faut introduire une notion ad-hoc (strictly redundant), mais je ne comprends pas pourquoi c'est utile. D'autre part il est clair que les clauses redondantes par rapport à $\projection{S}{\X}$ ne sont pas dans $S'$.}
in $S'$ are  \eflat{$\hypcst$} and of the form $c \iseq a' \vee a \niseq b \vee C'$, where:
\begin{enumerate}
\item $C'$ is positive,
\item $c \succ a'$, $a' \succeq a$ and $a' \succeq b$,
\item $\spc{c} \models C' \vee a\iseq a' \vee b\iseq a'$,
\item $\projection{S}{\X}$ contains a positive \eflat{$\hypcst$} clause $C'' \subseteq
c \iseq a \vee c \iseq b \vee D$ of type \ref{projext:set}  such that
$\{ D \} \cup \spc{c} \models C'$ and $D \preceq C'$.
\end{enumerate}
This immediately implies that $\Box \not \in S'$, hence that
$\Box \in \projection{S}{\X}$.
The proof is by structural induction on $S'$. Let $C \in S'$. Note that $C$ cannot be redundant in $\projection{S}{\X}$, by definition of $S'$ since the conclusion of an inference rule cannot be greater than its maximal premise.
\begin{itemize}
\item{Assume that $C$ is derived by the Reflection inference
    rule. Then, since $\projection{S}{\X}$ is weakly saturated under
    Reflection, the parent of $C$ must occur in $S'$, hence by the
    induction hypothesis, it must be of the form $c \iseq a' \vee a
    \niseq b \vee C'$, where $\projection{S}{\X}$ contains a clause
    $C'' \subseteq c \iseq a \vee c \iseq b \vee D$ such that $\{ D \}
    \cup \spc{c} \models C'$, $\spc{c} \models C' \vee a \iseq a'$ and
    $D \preceq C'$.
By definition of the Reflection rule we have $a = b$ and by Proposition \ref{prop:simulatefact} 
the clause $c \iseq a \vee D$ is necessarily redundant in $\projection{S}{\X}$.
But $C = c \iseq a' \vee C'$ is redundant
in $\{ c \iseq a \vee D \} \cup \spc{c}$ by Condition $3$ above, since $a=b$. Therefore, $C$ is redundant in $\projection{S}{\X}$, which is impossible.}

\item{Assume that $C$ is derived by Factorization. Then
$C$ is of the form $c \iseq a \vee a \not \iseq b \vee C'$, and its parent is
 $c \iseq a \vee c \iseq b \vee C'$. Note that this parent clause must be positive, otherwise
 $c \iseq a$ would not be selected, and that it is of type \ref{projext:set}.
 Thus, it cannot occur in $S'$, and $c \iseq a \vee c \iseq b \vee C' \in \projection{S}{\X}$. It is simple to verify that the induction hypothesis holds on $C$.}

\item{Assume that $C$ is generated by a Superposition from $C_1$ into $C_2$. Then one of the premises is necessarily in $S'$, and by the induction hypothesis, it contains a negative literal. Since a positive literal is selected in the first premise of the inference rule, we deduce that $C_1 = a\iseq d \vee C_1'$, where $a\succ d$, $C_2 = c\iseq a'\vee a\niseq b\vee C_2'$, and $C = c\iseq a'\vee b\niseq d\vee C_1'\vee C_2'$. Note that $C_1$ must be of type \ref{projext:set}; furthermore, $a\neq b$, since otherwise the Reflection rule would apply upon $C_2$,  $c\iseq a'\vee C_2'$ would be redundant in $\projection{S}{\X}$ and so would $C$.
We prove that $C$ verifies the induction hypothesis.
\begin{enumerate}
\item Since $C_1$ is a positive clause and $C_2'$ is positive by the induction hypothesis, it is clear that $C_1'\vee C_2'$ is positive.
\item Since $a'\succeq a \succ d$, we have $c \succ a'$, $a' \succeq b$ and $a' \succeq d$.
\item By the induction hypothesis, $\spc{c} \models C_2'\vee a\iseq a' \vee b\iseq a'$. Since $c\succ a' \succeq a \succ d$, \nikonew{j'ai remplacé $c'$ par $c$ } we deduce that $C_1 \in \spc{c}$, and therefore $\spc{c} \models C_2'\vee C_1'\vee d\iseq a'\vee b\iseq a'$.
\item By the induction hypothesis, there is a positive clause $C_2''
  \in \projection{S}{\X}$ of type \ref{projext:set} such that $C_2''
  \subseteq c \iseq a \vee c \iseq b \vee D_2$, where $\{ D_2 \} \cup
  \spc{c} \models C_2'$ and $D_2 \preceq C_2'$. If $C_2''$ does not
  contain literal $c\iseq a$, then the proof is immediate, by letting
  $C'' \isdef C_2''$ and $D \isdef D_2$. Otherwise, $C_2''$ is of the
  form $c\iseq a \vee E$, where $E\subseteq c\iseq b\vee D_2$. If $a$
  is \fixed in $c\iseq a\vee E$, then by Lemma \ref{lem:composeq},
  there is a clause in $\projection{S}{\X}$ that is contained in
  $c\iseq d \vee C_1' \vee E \subseteq c\iseq b\vee c\iseq d \vee C_1'
  \vee D_2$, and the proof is completed. Otherwise, by Proposition
  \ref{prop:fixedbis}, since $C_2''$ is of type \ref{projext:set}, \nikonew{J'ai ajouté ce qui précède pour que le lecteur voie à quoi sert cette hypothèse.}  $\projection{S}{\X}$ contains a clause $c\iseq
  d \vee E'$, where $E'$ is obtained from $E$ by replacing some
  occurrences of $a$ by $d$.
  Since $E \subseteq c\iseq b\vee D_2$ and $a\neq b,c$, we deduce that $E'\subseteq c\iseq b\vee D_2'$, for a clause $D_2'$ obtained from $D_2$ by replacing some occurrences of $a$ by $d$. Since $a\iseq d\vee C_1' \in \spc{c}$, we deduce that $\set{D_2'} \cup \spc{c} \models C_1'\vee D_2$, hence $\set{D_2'} \cup \spc{c}  \models C_1'\vee C_2'$. Now $a\succ d$, so that $D_2'\preceq D_2\preceq C_2'$, and the clause $d\iseq c\vee E' \subseteq d\iseq c\vee b\iseq c\vee D_2'$ \nikonew{J'ai remplacé $E$ par $E'$.} fulfills the required property.
\end{enumerate}
}
 \end{itemize}

\section{Proof of Lemma \ref{lem:term}}

\label{ap:term}

The proof is based on the following results:

\begin{newprop}\label{prop:redunif}
  If $(\sigma, \E)$ is a most general \Aunifier of $t\iseq s$, then for  all $\X$ such that $\E \subseteq \X$, \nikonew{J'ai mis simplement $\E$ au lieu de $\E^+$, puisqu'il est clair que la partie \Aset d'un most general \Aunifier est forcément positive (même si les littéraux négatifs sont autorisés dans la définition des substitutions: si il y a des littéraux négatifs, on les enlève et on a encore un \Aunifier)}\MnachowithAnswer{tout à fait d'accord, selon moi il faudrait dire que les Asets des A-unificateurs sont toujours positifs.}{j'ai modifié la définition des substitutions}
  $\red{t}{\X}$ and $\red{s}{\X}$ are unifiable, and $\red{\sigma}{\X}$ is a most general unifier of $\red{t}{\X} \iseq \red{s}{\X}$.
\end{newprop}
\begin{proof}
  This is because if $(S,\theta, \X)$ and $(S',\theta', \X')$ are
  \Aunification problems such that $(S,\theta,\X) \rightarrow
  (S',\theta', \X')$, then for all {\Aset}s $\Y$ such that
  ${\X'}\subseteq \Y$, \nikonew{même chose ici, $\X'$ au lieu de ${\X'}^+$.} we have $\red{S}{\Y} \rightarrow
  \red{S'}{\Y}$ for the standard unification rules. The proof follows
  by a straightforward induction.
\end{proof}
Since terms that are $\equivto$-equivalent
\nikonew{au lieu de ``elements in $\hypcst$''} cannot be distinguished \nikonew{au lieu de ``compared'', plus précis, parce qu'il ne suffit pas dire que les termes ne sont pas comparables entre eux (ce qui est vrai pour $\alwayssgreater{}$ aussi) mais il faut que la réécriture ne change pas la relation par rapport aux autres termes.}  by $\stronglysmaller{}$ and $\strsel$, we have the following result.  \nikonew{ajout}
\begin{newprop}\label{prop:redsel}
\nikonew{modifs ici: je relie directement $\sel$ et $\strsel$.}
  Let $C = \ccl{t\bowtie s \vee D}{\X}$ be an \Aclause, where $t\bowtie s \in \sel(t\bowtie s \vee D)$ and $t\not\alwayssmaller{} s$.
  Let $\sigma$ be a ground $\X$-pure substitution of domain $\var(\X)$. If $t\sigma \equivto t'$, then
   $t' \not \stronglysmaller s\sigma$ and $(t' \bowtie s\sigma) \in \strsel(t' \bowtie s\sigma \vee D\sigma)$.
   %
\end{newprop}
\begin{newprop}\label{prop:unifsig}
  Let $\mu$ be an m.g.u.\ of $t \iseq s$.
  \begin{itemize}
  \item If $\mu_1,\mu_2$ are such that $\dom(\mu_1)\cap\dom(\mu_2) = \emptyset$ and $\mu = \mu_1\mu_2$, then $\mu_2$ is an m.g.u.\ of $t\mu_1\iseq s\mu_1$.
  \item Let $\sigma$ be a substitution \nikonew{au lieu de permutation} such that $\dom(\sigma)
    \subseteq \var(t\iseq s)$ and $\dom(\sigma) \cap \dom(\mu) =
    \emptyset$. Then the restriction of $\mu\sigma$ to $\dom(\mu)$ is
    an m.g.u.\ of $t\sigma \iseq s\sigma$.
  \end{itemize}
\end{newprop}
\begin{proof}
  Since $t\mu_1\mu_2 = t\mu = s\mu = s\mu_1\mu_2$, it is clear that $t\mu_1$ and $s\mu_1$ are unifiable. If $\delta$ is a unifier of \nikonew{au lieu de m.g.u.\ sinon il me semble qu'il faudrait dire des choses en plus ensuite, mais ça n'est pas la peine} of $t\mu_1\iseq s\mu_1$, then $t\mu_1\delta = s\mu_1\delta$, hence $\mu_1\delta$ is a unifier of $t\iseq s$, and is therefore an instance \nikonew{au lieu de renaming} of $\mu = \mu_1\mu_2$, thus $\delta$ is an instance of $\mu_2$. This proves that $\mu_2$ is an m.g.u.\ of $t\mu_1\iseq s\mu_1$.

  The second point is a consequence of the fact that for any unification problem, if $S\rightarrow S'$, then $S\sigma\rightarrow S'\sigma$. The result is proved by induction on the transformation of the unification problem $\set{t\iseq s}$.
\end{proof}

  We prove that if the $\hypcst$-Superposition rule applied to
 $C,D$
 generates $E$,
 then for all
  $E'\in \projectionter{E}{\U}$,
there exists $C'\in
  \projectionter{C}{\U}$
  and $D' \in
  \projectionter{D}{\U}$
  such that $E'$ can be
  derived from $C',D',\U$ by \SP.
The proof for the other inference rules is
  similar. We let\Mnacho{remplacé $C'$ etc par $C_1$ etc pour être cohérent avec la phrase précédente et propagé les modifs.}
  \begin{eqnarray*}
    C & = & \ccl{u\iseq v \vee C_1}{\X},\\
    D & = & \ccl{t\bowtie s \vee D_1}{\Y},\\
    E & = & \ccl{(t[v]_p  \bowtie s \vee C_1\vee D_1)\mu}{\calZ},
  \end{eqnarray*}
  where $(\mu, \E)$ is an $(\X\cup \Y)$-pure most general \Aunifier of $u\iseq t|_p$ and $\calZ = (\X\cup \Y\cup \E)\mu$. Up to a renaming, we may assume that $\var(\calZ) \subseteq \var(\X\cup\Y)$, so that for all $x \in \dom(\mu) \cap \var(\X\cup\Y)$, $x\mu \in \hypcst \cup \var(\X \cup \Y)$.
  Let $E' \in \projectionter{\set{\ccl{E}{\calZ}}}{\calZ}$, and let
  $\sigma$ be the $\calZ$-pure substitution of domain $\var(\calZ)$
  such that $\red{\sigma}{\U} = \sigma$, $\calZ\sigma \subseteq \U$ and $E' = E\sigma$.
  We let $\C\isdef
  \dom(\mu) \cap \var(\X\cup \Y)$ and define $\mu_1$ as the
  restriction of $\mu$ to $\C$ and $\mu_2$ as the restriction of $\mu$
  to $\dom(\mu)\setminus \C$, so that $\mu = \mu_1\uplus
  \mu_2$. Consider the substitution $\delta \isdef \mu_1\sigma$. It is
  clear that $\delta$ is a ground $(\X\cup\Y)$-pure substitution of
  domain $\var(\X\cup\Y)$, and that
  $\X\mu_1,\Y\mu_1 \subseteq \calZ\sigma \subseteq \U$, therefore, $C'\isdef C\delta
  \in \projectionter{\set{\ccl{C}{\X}}}{\U}$ and $D'\isdef
  D\delta \in \projectionter{\set{\ccl{D}{\Y}}}{\U}$. The
  clause $C'$ is of the form $u_1\iseq v_1 \vee C_1'$, and the clause
  $D'$ of the form $t_1\bowtie s_1 \vee D_1'$, where:
  \begin{itemize}
  \item$u_1 = u\delta$, $v_1 = v\delta$ and $C_1' = C_1\delta$,
  \item $t_1 = t\delta$, $s_1 = s\delta$ and $D_1' = D_1\delta$.
  \end{itemize}
  Let $t_1'\isdef
  \red{{t_1}}{\U}|_p$   and $u_1' \isdef \red{u_1}{\U}$.
  By Proposition \ref{prop:redsel},  $t_1 \bowtie s_1$ and $u_1 \iseq v_1$ are selected in $C'$ and $D'$ respectively, and we have $t_1 \not \stronglyssmaller s_1$, $u_1 \not \stronglyssmaller v_1$. Thus, there is an \SP-derivation from
  $\{ C' \} \cup \U$  that generates the clause $u_1' \iseq v_1
  \vee C_1'$, and an \SP-derivation from $\{ D' \} \cup \U$ that generates
  $t_1[t_1']_p \bowtie s_1 \vee D_1'$: it suffices to use repeated applications of the \nikonew{note cosmétique: ici et à d'autres endroits je ne sais pas si il faut mettre une majuscule à ``superposition'' (dans superposition calculus, superposition rule etc.). Tu me diras ce que tu préfères \smiley)}\MnachowithAnswer{\smiley j'ai essayé de mettre 'Superposition rule' et 'by superposing $x$ into $y$', mais je ne sais pas si c'est la meilleure décision}{ça me va \smiley J'ai mis des majuscules partout aux noms des règles} Superposition rule from equations in $\U$ to replace every constant $a$ occurring in $u_1$ or $t_1|_p$ by $\red{a}{\U}$. Note that $u_1' \iseq v_1$ and
  $t_1[t_1']_p \bowtie s_1$ are both selected and that $u_1' \not \stronglyssmaller v_1$ and $t_1[t_1']_p \not \stronglyssmaller s_1$.

\MnachowithAnswer{proposition: remplacer $\alpha$ et $\beta$ par $\gamma$ et $\nu$ (j'ai fait des commandes si ça ne te convient pas)}{ok}
We prove
  that $t_1'$  and $u_1'$ are unifiable.
  For $i = 1,2$, let $\fstsub_i\isdef \red{\mu_i}{\U}$, and let
  $\fstsub \isdef \fstsub_1\uplus \fstsub_2$. By Proposition
  \ref{prop:redunif}, since $(\mu,\E)$ is a most general \Aunifier of
  $t|_p\iseq s$ and $\E \subseteq \U$, $\fstsub$ is a most general
  unifier of $\red{(t|_p)}{\U} \iseq \red{u}{\U}$. By
  Proposition \ref{prop:unifsig}, $\fstsub_2$ is an m.g.u.\ of
  $\red{(t|_p)}{\U}\fstsub_1 \iseq \red{u}{\U}\fstsub_1$, and the
  restriction $\smgu$ of $\fstsub_2\sigma$ to $\dom(\fstsub_2)$ is an m.g.u.\ of
  $\red{(t|_p)}{\U}\fstsub_1\sigma \iseq
  \red{u}{\U}\fstsub_1\sigma$. But we have
  \[\red{(t|_p)}{\U}\fstsub_1\sigma\ =\ \red{(t\mu_1\sigma)|_p}{\U}\ =\ \red{(t\delta)|_p}{\U}\ =\ {t_1'},\]
  and similarly, $\red{u}{\U}\fstsub_1\sigma = u_1'$. Since
  $t_1'$ and $u_1'$ are unifiable with m.g.u.\ $\smgu$, the
  Superposition rule applied to $u_1' \iseq v_1 \vee C_1'$ and  $t_1[t_1']_p \bowtie s_1 \vee D_1'$ generates the clause $F\isdef (t_1[v_1]_p \bowtie s_1\vee C_1'\vee D_1')\smgu$, and:
  \begin{eqnarray*}
    F & = & (t_1[v_1]_p \bowtie s_1\vee C_1'\vee D_1')\smgu\\
    & = & (t[v]_p\delta \bowtie s\delta \vee C'\delta\vee D'\delta)\smgu\\
     & = & (t[v]_p\bowtie s\vee C'\vee D')\mu_1\sigma\smgu.
  \end{eqnarray*}
  We now prove that for any variable $x$, we have $x\mu_1\sigma\smgu = x\mu\sigma$. First assume that $x\notin \dom(\mu_1)$. If $x\in \var(\X\cup\Y)$, then necessarily $x\in \var(\calZ)$, and therefore, $x\mu = x$ and $x\in \dom(\sigma)$. Thus, $x\mu_1\sigma\smgu = x\sigma\smgu = x\sigma = x\mu\sigma$. Otherwise, since $\dom(\sigma) \subseteq \var(\X\cup\Y)$, necessarily $x\sigma = x$ and $x\mu_1\sigma\smgu = x\sigma\smgu = x\smgu$. If $x\in \dom(\smgu)$ then $x\smgu = x\mu\sigma$ by definition of $\smgu$, otherwise, since $x\notin \dom(\mu_1) \uplus \dom(\mu_2) = \dom(\mu)$, we deduce that $x\smgu = x = x\mu = x\mu\sigma$.
  Now assume that $x\in \dom(\mu_1)$. Then $x\mu_1 = x\mu$, and if $x\mu \in \hypcst$, then $x\mu_1\sigma\smgu = x\mu = x\mu\sigma$. Otherwise $ x\mu \in \var(\calZ) = \dom(\sigma)$, hence $x\mu\sigma\fstsub = x\mu\sigma$.

For the second part of the lemma, let $E \isdef \ccl{E'}{\calZ}$ and suppose that $\projectionter{E}{\U}$ contains a clause $E'\gamma'$ (with $\red{\gamma'}{\U} = \gamma'$) that is \stronglyredundant in $\projectionter{S}{\U}$.
 Let $\sigma$ be a ground substitution of the variables in $\ccl{E'}{\calZ}$ such that $\calZ\sigma \subseteq \U$.
 We show that $\ccl{E'}{\calZ}\sigma$ is \redundant in $S$. We assume, w.l.o.g., that $\sigma = \red{\sigma}{\X}$.
Let $\gamma$ and $\theta$ be the restrictions of $\sigma$ to $\var(\X)$ and $\dom(\theta) \setminus \var(\X)$ respectively. By definition we have $\dom(\gamma) = \var(\X) = \dom(\gamma')$, hence $E'\gamma' \equivto E\gamma$.
Since $E'\gamma'$ is \stronglyredundant in $\projectionter{S}{\U}$ we deduce that
$E'\gamma\theta = E'\sigma$ is \redundant in $S$. Since $S$ is a set of standard clauses, this entails that $E\sigma$ is also \redundant.

\end{document}